\documentclass[a4paper,12pt,reqno]{amsart}
\usepackage[english]{babel}
\usepackage[utf8]{inputenc}
\usepackage[T1]{fontenc}
\usepackage{amssymb,amsmath,amsthm,bm} 
\usepackage{enumerate}
\usepackage{latexsym}
\usepackage{mathtools}
\usepackage{mathrsfs}
\usepackage{fancyhdr}
\usepackage{setspace}
\usepackage{color} 
\usepackage{graphicx} 
\usepackage{subfigure}
\usepackage{caption}

\usepackage{geometry}
\usepackage[pdfpagelabels,plainpages=false]{hyperref}
\hypersetup{colorlinks=false}
\usepackage{bm}

\numberwithin{equation}{section}

\DeclarePairedDelimiter{\abs}{\lvert}{\rvert}   
\DeclarePairedDelimiter{\norma}{\lVert}{\rVert}

\DeclareMathOperator{\diag}{diag}

\newcommand{\R}{ \mathbb{R}}
\newcommand{\Hil}{\mathcal{H}}

\newcommand{\four}{\mathcal{F}}
\newcommand{\n}{\noindent}
\newcommand{\f}{\frac}
\newcommand{\vs}{\vspace{0.5cm}}
\newcommand{\vsa}{\vspace{0.2cm}}

\newcommand{\ba}{\begin{eqnarray}} 
\newcommand{\ea}{\end{eqnarray}}

\newcommand{\be}{\begin{equation}}
\newcommand{\ee}{\end{equation}}

{\left\lbrace\begin{array}{@{}l@{}}}%
{\end{array}\right.}

\theoremstyle{plain} 
\newtheorem{lemma}{Lemma}
\newtheorem{theorem}{Theorem}
\newtheorem{proposition}{Proposition}

\theoremstyle{remark}
\newtheorem{remark}{Remark}
 
\theoremstyle{definition}
\newtheorem{definition}{Definition}

\title{Efimov effect for a three-particle system with two identical  fermions}

\author[]{Giulia Basti}
\address{Dipartimento di Matematica G. Castelnuovo, Sapienza Universit\`a di Roma,  Piazzale  Aldo Moro, 5, 00185 Roma, Italy}
\email{basti@mat.uniroma1.it}

\author[]{Alessandro Teta}
\address{Dipartimento di Matematica G. Castelnuovo, Sapienza Universit\`a di Roma,  Piazzale  Aldo Moro, 5, 00185 Roma, Italy}
\email{teta@mat.uniroma1.it}

\date{}

\makeatletter
\let\orgdescriptionlabel\descriptionlabel
\renewcommand*{\descriptionlabel}[1]{%
  \let\orglabel\label
  \let\label\@gobble
  \phantomsection
  \edef\@currentlabel{#1}%
  \let\label\orglabel
  \orgdescriptionlabel{#1}%
}
\makeatother

\begin{document}


\vs


\begin{abstract}

We consider a three-particle quantum system in dimension three composed of two identical fermions of mass one and  a different particle of mass $m$. The particles interact via two-body short range potentials. We assume that  the  Hamiltonians of all the two-particle subsystems do not have bound states with negative energy and, moreover,  that the Hamiltonians of the two subsystems made of a fermion and the different particle have a zero-energy resonance. 
Under these conditions and for $m<m^* = (13.607)^{-1}$, we give a rigorous proof of the occurrence of the Efimov effect, i.e., the existence of infinitely many negative eigenvalues for the three-particle Hamiltonian $H$. More precisely, we prove that  for $m>m^*$ the number of negative eigenvalues of $H$ is finite and   for $m<m^*$ the number $N(z)$ of negative eigenvalues of $H$ below $z<0$ has the asymptotic behavior $N(z) \sim \mathcal C(m) |\log|z||$ for $z \rightarrow 0^-$. Moreover, we give an upper and a lower bound for the positive constant $\mathcal C(m)$.

\end{abstract}
\maketitle

\vs

\section{Introduction}

Efimov effect is a remarkable physical phenomenon occurring in three-particle quantum systems in dimension three. It was  discovered by Efimov in 1970  (\cite{efimov1}, \cite{efimov2}) and it  consists in the following. Let us assume that the  particles interact via two-body  short range potentials, the two-particle subsystems do not have bound states and at least two of them exhibit a zero-energy resonance.
 Then the Hamiltonian describing the three-particle system has infinitely many negative eigenvalues $E_n$ accumulating at zero. Moreover, the eigenvalues satisfy the asymptotic geometrical law
\be\label{gela}
\f{E_{n+1}}{E_n} \; \rightarrow \; e^{-\f{2\pi}{s_0}}\,, \quad\quad \text{for} \quad n\rightarrow \infty
\ee
where the parameter $s_0>0$ depends only on the mass ratios and, possibly, on the statistics of the particles. 
The  three-particle bound states (or trimers) associated to the eigenvalues $E_n$  are characterized by a size much larger than the range of the two-body potentials. They are determined by a long range, attractive effective interaction of kinetic origin which is produced by the resonance condition and it is independent of the details of the two-body potentials. According to an intuitive physical picture,   one can say that in a trimer     
 the attraction between two particles is  mediated by the third one, which is moving back and forth between the two. Note that  the attraction disappears if the two-body potentials become more attractive causing the destruction of the zero-energy resonance. 
 
\n
We emphasize that Efimov effect describes  a universal low-energy behavior of the three-particle system.  As a consequence of this  universality character,  the effect can be realized and observed   in various physical contexts (e.g., in atomic, molecular, nuclear or condensed matter physics) and this fact has motivated a large number of experimental and theoretical works  published on the subject in recent years (see, e.g., the reviews \cite{braaten}, \cite{naidon}).

\n
The original Efimov's physical argument  is based on the replacement of the two-body potential with a boundary condition, which is essentially equivalent to consider a two-body zero-range interaction, and on the introduction of hyper-spherical coordinates. If the resonant condition is satisfied, in these coordinates the problem become separable  and in the equation for the hyper-radius $R$ the long range, attractive effective  potential $- (s_0^2 +1/4)/R^2$ appears. The behavior for small $R$ of this potential is too singular and an extra boundary condition at short distance  must be imposed to restore self-adjointness. After this ad hoc procedure, one obtains the infinite sequence of negative eigenvalues satisfying the law \eqref{gela} as a consequence of the large $R$ behavior of the effective potential.

\n
The first mathematical result on the Efimov effect was obtained by Yafaeev in 1974 (\cite{yafaev}). He studied a symmetrized form of the Faddeev equations  for the bound states of the three-particle Hamiltonian and proved the existence of an infinite number of negative eigenvalues. In 1993 Sobolev (\cite{sobolev}) used a slightly different symmetrization of the equations and proved the asymptotics 
\be\label{asob}
\lim_{z \rightarrow 0^-} \f{N(z)}{ |\log |z||} = \f{s_0}{2\pi}  
\ee
 where $N(z)$ denotes the number of eigenvalues smaller than $z<0$. Note that \eqref{asob} is consistent with the law \eqref{gela}. In the same year Tamura (\cite{Tam2}) obtained the same result under more general conditions on the two-body potentials.  Other mathematical proofs of the effect were obtained by  Ovchinnikov and Sigal in 1979 (\cite{sigal}) and Tamura in 1991 (\cite{Tam1}) using a variational approach based on the Born-Oppenheimer approximation (see also the related result in \cite{simon}). For more recent results on the subject see  \cite{gridnev1}, \cite{gridnev2}.  It is worth mentioning that a mathematical proof of the geometric asymptotic law \eqref{gela}  is still lacking (see the conjecture discussed in \cite{ahkw}). We also mention that for three identical fermions in dimension three the Efimov effect is absent (\cite{VZ})

\n
In this paper we study the case of a three-particle system in dimension three  composed of two identical fermions with mass one and a different particle with mass $m$. In the physical literature  such a system has been extensively studied (see, e.g., \cite{castin}, \cite{naidon} and references therein) and it is known  that the Efimov effect can be present with  some peculiar features. Indeed,    the effect is present only for $m < m^* = (13.607)^{-1}$ and,  considering  a partial wave decomposition, it takes place only in the subspaces corresponding to the    \hspace{0.01cm} ``odd waves'' (contrary to the case of identical bosons or distinguishable particles where the effect takes place in the ``s-wave'' subspace). Following the approach based on the analysis of the Faddeev equations, we give a mathematical proof of these facts. More precisely, we assume that: a)  the two-body potentials are short  range, rotationally invariant, non positive; b) the Hamiltonians of all   the two-particle subsystems are positive; c)  the Hamiltonian of the subsystems composed of the particle with mass $m$ and a fermion has a zero-energy resonance. Then we prove that  the Hamiltonian of the three-particle system has a finite number of negative eigenvalues for $m>m^*$ and  an infinite number of negative eigenvalues accumulating at zero for $m<m^*$. We also prove the asymptotic behavior \eqref{asob}, where the constant at the right hand side, which in our case is denoted by $\mathcal C(m)$, depends only on the mass $m$ and it is estimated from below and from above.

\n
We note that, under our assumptions, the interaction potential between the two fermions does not produce zero-energy resonance and therefore it plays no role in the occurrence of the Efimov effect.

\n
The method of the proof follows the line of reasoning of \cite{sobolev}, with the modifications required  to take into account of the peculiarity of our system. In particular, we formulate the eigenvalue problem $H\psi=z\,\psi,\;$ $z<0,\;$ for the three-particle Hamiltonian in terms of symmetrized Faddeev equations $\Psi = \boldsymbol{A} (z) \Psi$, 
where   $\boldsymbol{A}(z)$ is a $2\times 2$ matrix (compact) operator, and we  prove that $ N(z)=n(1,\boldsymbol{A}(z))$, 
where the right hand side is the number of eigenvalues of $\boldsymbol{A}(z)$ larger than one. Then the problem is reduced to the study of the asymptotics of $n(1,\boldsymbol{A}(z))$ for $z \rightarrow 0^-$. The presence of the zero-energy resonance for the subsystems composed of the particle with mass $m$ and a fermion determines a singular behavior (and a  lack of compactness) of $\boldsymbol{A}(z)$ for $z=0$ and this is the reason for the possible divergence of $n(1,\boldsymbol{A}(z))$ for $z \rightarrow 0^-$.  Through some successive steps, we single out such singular behavior  neglecting operators which, for $z\leq 0$, are compact and continuous in $z$.   At the end we find that the asymptotics for $z \rightarrow 0^-$ of $n(1,\boldsymbol{A}(z))$  reduces to the asymptotics for $R\rightarrow \infty$ of an operator $S_R$ which has an explicit form. By a direct analysis of such an operator, we conclude the proof of the main result in the two cases $m<m^*$ and $m>m^*$.

\n
The paper is organized as follows.

\n
In section 2 we describe the three-particle model, formulate our assumptions on the interaction potentials and state the main result. 
In section 3 we briefly recall some results on the low-energy behavior of a two-particle Hamiltonian in presence (or in absence) of a zero-energy resonance. 
In section 4 we introduce the symmetrized Faddeev equation for the bound states of our three-particle Hamiltonian and we  prove that $ N(z)=n(1,\boldsymbol{A}(z))$. 
In section 5 we characterize the leading term of $\boldsymbol{A}(z)$ for $z \rightarrow  0^-$, neglecting operators which, for $z \leq 0$, are compact and continuous in $z$.  
In section 6 we prove the main result exploiting the asymptotic behavior of the leading term.

\vs\vs
\section{Notation and main result}

We consider a quantum system composed by two identical fermions of unitary mass  and a different particle of mass $m_3=m.$ Let $x_1,x_2\in \R^3$ denote the coordinates of the fermions and $x_3\in\R^3$ the coordinates of the third particle. The state of the system is then described by  a wave function $\psi\in L^2(\R^9) $ which satisfies the symmetry condition $ \psi(x_1,x_2,x_3)=-\psi(x_2,x_1,x_3)$ and the Hamiltonian is typically of the form 
\be
	\tilde{H}= -\f{1}{2} \Delta_{x_1}  -\f{1}{2} \Delta_{x_2} -\frac{1}{2m}\Delta_{x_3}+v_{12}(x_1 - x_2) + v_{23}(x_2 - x_3) + v_{31}(x_3 - x_1)
\ee
where $\Delta_{x_i}$ denotes the laplacian with respect to the coordinates of the i-th particle and $v_{\alpha}$,  $\alpha\in\{12,23,31\}$ is the two-body real-valued potential associated to 
 the pair of particles $\alpha$. Due to the simmetry constraint,  we have $v_{23}=v_{31}:=v$. 
 
\n
We introduce the coordinates $(R, x_{\alpha},y_{\alpha})$, where $R$ is the coordinate of the center of mass and 
 $(x_{\alpha},y_{\alpha})$ is any pair of Jacobi coordinates, e.g.,  for $\alpha = 12$ one has  
\begin{equation}
x_{12}= x_1 - x_2, \qquad \qquad 
y_{12}= x_3 - \f{ x_1 + x_2}{2}.  
\end{equation}
In such coordinates  the Hamiltonian takes the form
\be
\tilde{H} = - \f{1}{2(m+2)} \Delta_{R} - \f{1}{2 \mu_{\alpha}} \Delta_{x_{\alpha}} - 
\f{1}{2 n_{\alpha}} \Delta_{y_{\alpha}} + \sum_{\beta} v_{\beta} (x_{\beta})
\ee
where
\be
	\begin{aligned}
		\mu_{12}&=\frac{1}{2} &\qquad \mu_{23}&=\mu_{31}= \mu=\frac{m}{m+1}\\
		n_{12}&=\frac{2m}{m+2} &\qquad n_{23}&=n_{31}=n=\frac{m+1}{m+2}
	\end{aligned}
\ee

\n
and $x_{\beta}$, for $\beta \neq \alpha$, is expressed in terms of $(x_{\alpha}, y_{\alpha})$.
\n
Moreover, it is convenient to extract the center of mass motion and to study the problem in momentum space. Let   $(k_\alpha,p_\alpha)$ be the pair of  variables conjugate with respect to the Jacobi coordinates $(x_{\alpha},y_{\alpha})$.  Denoting with $k_i$ the conjugate variable of $x_i$, they are explicitely defined by
\begin{equation}\label{eq:k,p}
	\begin{aligned}
		k_{12}&=\displaystyle{\frac{k_1-k_2}{2}} &\qquad p_{12}&=\displaystyle{\frac{m(k_1+k_2)-2k_3}{m+2}}\\		
		k_{23}&=\displaystyle{\frac{mk_2-k_3}{m+1}} &\qquad p_{23}&=\displaystyle{\frac{(k_2+k_3)-(m+1)k_1}{m+2}}\\
		k_{31}&=\displaystyle{\frac{k_3-mk_1}{m+1}} &\qquad p_{31}&=\displaystyle{\frac{(k_1+k_3)-(m+1)k_2}{m+2}}.
	\end{aligned}
\end{equation}
Then, for any choice of pair $\alpha$,  the Hamiltonian of the system can be written as

\be
H=H_0+\sum_{\beta}V_\beta
\ee
 where 
 \be
 \displaystyle{H_0=\frac{k_\alpha^2}{2\mu_{\alpha}}+\frac{p_\alpha^2}{2n_\alpha}}
 \ee
  is the free Hamiltonian and
\be\label{intfour}
	(V_\beta\psi) (k_\beta,p_\beta)=\frac{1}{(2\pi)^{3/2}}\int dk\, \hat{v}_\beta(k-k_\beta)\psi(k,p_\beta)
\ee
describes each interaction term. In~\eqref{intfour} and in the following we denote  by $\hat{f}$  the Fourier transform of $f.$ 
Taking into account of the definitions given in  \eqref{eq:k,p},  the symmetry constraint reduces to  $\psi(k_{23},p_{23})=-\psi(-k_{31},p_{31})$ or equivalently $\psi(k_{12},p_{12})=-\psi(- k_{12},p_{12}).$ Therefore, choosing for instance the coordinates $(k_{23},p_{23})$,  the Hilbert space of the system is
\be\label{eq:Hil}
	\Hil=\left\{\psi\in L^2(\R^6)\,|\, \psi(k_{23},p_{23})=-\psi\left(\frac{k_{23}}{m+1}-\frac{m(m+2)}{(m+1)^2}p_{23},-k_{23}-\frac{p_{23}}{m+1}\right)\right\}.
\ee
During the proof  it will be useful to use also  the system of coordinates $(p_{23},p_{31}).$ Using the relations
\begin{equation}\label{eq:p,q}
	\begin{aligned}
	k_{23}&=-p_{31}-\frac{1}{m+1}p_{23}\\
	k_{31}&=p_{23}+\frac{1}{m+1}p_{31}\\
	\end{aligned}
\end{equation}
one finds that the free Hamiltonian can be rewritten as 

\be\label{eq:H_0(p)}
\displaystyle{H_0= \f{p^2_{23}}{2\mu} +\f{p^2_{31}}{2\mu} +\frac{p_{23}\cdot p_{31}}{m}}.
\ee
The exchange of the fermionic coordinates corresponds to the exchange of $p_{23}$ and $p_{31}$. Thus, using the coordinates $(p_{23},p_{31})$,  the  Hilbert space of the system can be equivalently written as
\be
	\Hil_1=\big\{\psi\in L^2(\R^6)\,|\, \psi(p_{23},p_{31})=-\psi(p_{31},p_{23}) \big\}.
\ee

\vsa

\n
Note that for $f \in \mathcal H$ we have $g \in \mathcal H_1$, where $
g(p_{23},p_{31})= f\!\left(\! -p_{31}-\f{p_{23}}{m+1}, p_{23}\!\right).$ 

\n
It is also useful to introduce the two-particle subsystems of our three-particle system,  described  by the following  Hamiltonians in $L^2(\R^3)$
\be\label{halfa}
h_{\alpha}= -\frac{1}{2\mu_{\alpha}}\Delta_{x}+v_{\alpha}(x)
\ee

\n
In order to formulate our main result we introduce below our assumptions on the two-body potentials, on the zero energy properties of the two-particle subsystems and on the mass ratio. 

\n 
Concerning the two-body potentials $v_{\alpha}$ and the Hamiltonians $h_{\alpha}$, we assume that the following conditions hold for any pair $\alpha$

\vsa
\begin{description}
	\item[\label{ass1}{$(A_1)$}] $\abs{v_{\alpha}(x)}\leq C(1+\abs{x})^{-b}$, with $b>3$
	\vsa
	\item[\label{ass2}{$(A_2)$}] $v_{\alpha}$ is spherically symmetric, that is $v_{\alpha}(x)=v_{\alpha}(\abs{x})$
	
	\vsa
	\item[\label{ass3}{$(A_3)$}] $v_{\alpha}\leq 0$    
	\vsa
	\item[\label{ass4}{$(A_4)$}] $h_\alpha \geq 0\;\;\;$ 
\end{description} 

\vsa

\n
A further important assumption is the presence of a zero energy resonance for the two-particle Hamiltonian $h_{23}=h_{31}$. Here we  recall  the definition of a zero-energy resonance,  while further comments and some useful results will be given in the next section. 

\n
Let us consider an Hamiltonian in $L^2(\R^3)$

\be\label{hu}
h=-\frac{1}{2\mu}\Delta_x+u(x)
\ee
with $\mu>0$ and $u$ a generic  potential  satisfying  assumptions \ref{ass1}, \ref{ass2} and \ref{ass3}. We denote by $g_0$ the integral operator with kernel 
\be
	g_0(x,x')=\frac{\mu}{2\pi|x-x'|}.
\ee
Then we have
\begin{definition}\label{defre}
 Zero is a (simple) resonance for $h$ if $1$ is a simple eigenvalue of the operator $|u|^{1/2}g_0|u|^{1/2}$ and the 	corresponding eigenfunction $\varphi$ satisfies $(|u|^{1/2},\varphi)\neq 0.$ 
\end{definition}

\n

\n
Let us comment on the above assumptions. We remark that if each $v_{\alpha}$ satisfies condition \ref{ass1} then, via Kato-Rellich theorem,  we have self-adjointness and lower boundedness of $H$ on the same domain of $H_0.$

\n 
Moreover, using the HVZ theorem, we know that the essential spectrum of $H$ is of the form $[l,+\infty)$, where $l$ is the lowest point of the spectra of the operators $h_{\alpha}$ describing the two-particle subsystems. Thus the assumption \ref{ass4} implies that on the right of $0$ the spectrum of $H$ contains only isolated eigenvalues with finite multiplicities.  

\n
The assumptions \ref{ass2}, \ref{ass3} are introduced to simplify the analysis. Note that  zero-energy resonance can not be present in a two-particle system if the interaction potential is positive or if the two particles are identical fermions and the potential is spherically symmetric 
 (see remark \ref{remfer} at the end of the next section).

\n
As we already pointed out in the introduction, a peculiar aspect of the Efimov effect in our fermionic system is that it only occurs for a certain range of values of the mass ratio $(0,m_*)$, which can be defined as follows. 
Let $\Lambda(m)$ be the following function of the mass
\begin{equation}\label{eq:Lambda}
	\Lambda(m)=\frac{2(m+1)^2}{\pi}\left(\frac{1}{\sqrt{m(m+2)}}-\arcsin\left(\frac{1}{m+1}\right)\right).
\end{equation}
It is easy to check that $\Lambda(m)$ is decreasing and
\begin{equation}
	\lim_{m\to 0}\Lambda(m)=+\infty,\qquad\lim_{m\to +\infty}\Lambda(m)=0
\end{equation}
Thus the following definition makes sense
\begin{definition}\label{def:m*}
The critical mass $m_*$ is the unique solution of the equation $\Lambda(m)=1.$
\end{definition}

\n
Note that for $m<m_*$ ($m>m_*$) we have $\Lambda(m)>1$ ($\Lambda(m)<1$). Moreover,  we stress that $m_*$ is the same mass threshold obtained   in the study of the corresponding system with point interactions (see, e.g., \cite{fincoteta}, \cite{CDFMT}, \cite{minlos}).

\n
We are now ready to formulate our main result.
\begin{theorem}\label{th:main1}
Let us assume that conditions \ref{ass1}, \ref{ass2}, \ref{ass3}, \ref{ass4} hold for any pair $\alpha$ and that  the two-particle Hamiltonian  $h_{23}=h_{31}$ has a zero energy resonance. Then the following holds.

\n
\begin{enumerate}[i)]
	\item For $m<m_*$  there exists a positive constant $\,\mathcal C(m)$  such that 
		\be\label{eq:main}
			\lim_{z\to 0^-}\frac{N(z)}{|\log|z||}=\mathcal C(m).
		\ee
		Moreover,  $\mathcal C(m)$ satisfies  
		\begin{equation}\label{eq:C(m)}
			\mathcal{C}_1(m) \leq  \mathcal{C}(m) \leq   \mathcal{C}_2(m)    		\end{equation}
	where $\mathcal{C}_1(m)$ is the unique positive solution of the equation $F_m(x)=1$, with 
	\be
	F_m(x):= \frac{m+1}{\sqrt{m(m+2)}}\int_0^1 dy\,y \frac{\sinh\left( \f{2\pi}{3} x \arcsin\left(\frac{y}{m+1}\right)\right)}{\sinh\left( \f{\pi^2}{3} x \right)\cos\left(\arcsin\left(\frac{y}{m+1}\right)\right)}\,, 
	\ee
	and
	\begin{equation} \label{eq:C(m)_upper}
	\mathcal{C}_2(m) = \frac{1}{4\pi\beta(m) }\log\left( \sqrt{ \f{2\pi}{3} } \alpha(m) \right)(l_0^2(m) +3l_0(m)+2)
\end{equation}
with
\begin{equation}\label{eq:alphabeta}
	\alpha(m):=\frac{(m+1)^{3/2}}{2\sqrt{\pi}\sqrt{m(m+2)}}\log^{1/2}\left(1+\frac{2}{m}\right),\quad \beta(m):=\frac{\pi}{2}-\arcsin\left(\frac{1}{m+1}\right)
\end{equation}
and
$l_0(m)$ is the largest odd integer smaller than $\pi \alpha(m)^2 - 1/2$.
		
	\item For $m>m_*$ the number of negative eigenvalues of $H$ is finite.
\end{enumerate}
\end{theorem}


\vs\vs
\section{Two-particle subsystems}
Here we recall some properties of the  two-body Hamiltonian operator $h$ in the position space  defined by~\eqref{hu},  under the hypothesis that the assumptions \ref{ass1},\ref{ass2},\ref{ass3}  and \ref{ass4} are satisfied. In particular, we are interested in the low energy behavior connected with the presence of zero-energy resonance (see definition~\ref{defre}). We first observe that 
if $\varphi$ is a square integrable solution of $|u|^{1/2}g_0|u|^{1/2}\varphi=\varphi$ then $\psi=g_0|u|^{1/2}\varphi$ satisfies $h\psi=0$ in the sense of distributions. Furthermore,   $\psi$ is an eigenfunction of $h$ with eigenvalue zero  if and only if $(|u|^{1/2},\varphi)=0$ (see, e.g., section 1 in \cite{Tam2}).  On the other hand,   if zero is a resonance for $h$ then there exists $\psi$ solution of $h\psi=0$ in the sense of distribution with $\psi\in L^2_\textup{loc}(\R^3)$ but $\psi \notin L^2(\R^3).$

\begin{remark}\label{rem:resonance}
We note that 	if $1$ is a multiple eigenvalue for $|u|^{1/2}g_0|u|^{1/2}$ then one can always find a $\varphi$, linear combination of eigenfunctions, such that $(u|^{1/2},\varphi)=0.$
\end{remark}

\begin{remark}\label{rem:s_wave}
The spherical symmetry of the potential $u$ implies that zero-energy resonance can only occur in \emph{s-wave} subspace.  Indeed, if $\varphi$ belongs to the subspace with angular momentum $l\geq 1$ then obviously $(|u|^{1/2},\varphi)=0.$
\end{remark}

\n
Let us introduce the resolvent of $h$ 

\be
r(z) = (h-z)^{-1}
\ee
which is  a bounded operator in $L^2(\R^3)$  for any  $z<0$ due to condition \ref{ass3},  the free resolvent 
\be
r_0(z) = \big(-\f{1}{2\mu} \Delta -z \big)^{-1} 
\ee
 (note that $r(0)=g_0$) and the operator

\be\label{w(z)}
	w(z)=I+|u|^{1/2}r(z)|u|^{1/2}. 
	\ee

\vsa

\n
Using the resolvent identity $r(z)=r_0(z)- r(z) \,u\, r_0(z)$ one verifies that

\be
w(z)=(I-|u|^{1/2}r_0(z)|u|^{1/2})^{-1}
\ee

\vsa
\n
The following lemma describes the behavior of the operator $w(z)$ in the case 1 is not an eigenvalue of $|u|^{1/2}g_0|u|^{1/2}.$
\begin{lemma}\label{lemma:no_res}
	Let us suppose that $1$ is not an eigenvalue of the operator $|u|^{1/2}g_0|u|^{1/2}.$ Then $w(z)$ is continuous in $z\leq 0.$
\end{lemma}

\n
Let us assume that  $h$ has a zero-energy resonance. Then we fix the following normalization condition on the eigenvector $\varphi$
\begin{equation}\label{eq:norm}
	(|u|^{1/2},\varphi)=(2\pi)^{3/2} ( \widehat{\varphi |u|^{1/2}})(0)=2^{1/4}\pi^{1/2}\mu^{-3/4}
\end{equation}
and we characterize the behavior of the  operator $w(z)$ when $z\rightarrow 0$.
\begin{lemma}\label{lemma:res}
	Suppose that $h$ has a zero-energy resonance.  If $z<0$ is small enough and $\delta<\text{min}\{1,b-3\}$ then
	\be
		w(z)=\frac{(\cdot,\varphi)\varphi}{|z|^{1/2}}+|z|^{-\frac{1-\delta}{2}}w^{(\delta)}(z)
	\ee
	where the operator $w^{(\delta)}(z)$ is  continuous in $z\leq 0.$ Moreover 
	\be
		(w(z))^{1/2}=\frac{(\cdot,\varphi)\varphi}{\norma{\varphi}|z|^{1/4}}+|z|^{-\frac{1-\delta}{4}}\tilde{w}^{(\delta)}(z)
	\ee
	where the operator $\tilde{w}^{(\delta)}(z)$ is  continuous in $z\leq 0.$
\end{lemma}

\n
For the proof of the previous lemmas we refer the reader to \cite{sobolev}.

\begin{remark}\label{remfer}
\n
Let us consider the Hamiltonians of the two-body subsystems defined in \eqref{halfa}.  We note that the operator $h_{12}$, due to the symmetry constraint, acts on 
\begin{equation}
L_{\textup{asym}}^2(\R^3)=\{\psi\in L^2(\R^3)\,|\, \psi(-x)=-\psi(x)\}.
\end{equation}
 Hence, by remark~\ref{rem:s_wave}, it cannot have zero energy resonance because its domain does not include s-wave functions. Moreover, using the assumptions \ref{ass2} and \ref{ass4}, $h_{12}$ cannot have zero as an eigenvalue (see e.g. \cite{yafaev_virt}). By  remark~\ref{rem:resonance}, we conclude that 1 is not an eigenvalue for the operator $|v_{12}|^{1/2}g_0|v_{12}|^{1/2}.$ Then  we can  apply  lemma~\ref{lemma:no_res} and, as we will see,  this implies that the potential between the two fermions does not play any role in the proof of the Efimov effect.

\n
On the other hand, as we have already underlined, the presence of a zero-energy resonance for  $h_{23}=h_{31}$  is a crucial ingredient of  the proof.
\end{remark}

\vs\vs

\section{Faddeev equations}
	Our proof of the occurrence of the Efimov effect is based on the analysis of the Faddeev equations (\cite{faddeev}), which are the three-body analogous of the Birman-Schwinger equation for the one body problem. For the convenience of the reader,  we recall here the derivation following the clear and simple presentation contained in  \cite{motovilov}. 

\n
Let us consider the eigenvalue equation for $H$
\be\label{eigH}
\big(H_0 + \sum_{\alpha} V_{\alpha} \big) \psi = z \, \psi\,, \;\;\;\;\;\;\;\;z<0
\ee
Note that the free resolvent
\be
R_0(z)= (H_0 -z)^{-1}
\ee

\vsa
\n
is bounded in $\mathcal H$ for any $z<0$, so we can equivalently write 
\be\label{LS}
\psi = - R_0(z) \sum_{\alpha} V_{\alpha} \psi
\ee
We decompose $\psi$ in the Faddeev components 
\be
\psi = \sum_{\alpha} \eta_{\alpha}
\ee
where, for each pair $\alpha$, from \eqref{LS} we have 
\be\label{eq:defetaa}
\eta_{\alpha} = - R_0(z)  V_{\alpha} \psi
\ee
In order to find the equations for $\eta_{\alpha}$, we introduce the following operators acting in the Hilbert space $L^2(\R^6)$
\be
H_{\alpha}=H_0 + V_{\alpha}
\ee
i.e., the Hamiltonian of the three-particle system with the interactions between the pairs $\beta$, with $\beta \neq \alpha$, removed, and its resolvent 

\be
R_{\alpha}(z)= (H_{\alpha} - z)^{-1}
\ee
which is bounded in $L^2(\R^6)$ for any $z<0$ by our assumptions on the potentials. Then we rewrite \eqref{eq:defetaa} in the form
\be
\eta_{\alpha}= - R_0(z) V_{\alpha} \sum_{\beta} \eta_{\beta}
\ee
or 
\be
\eta_{\alpha} +R_0(z) V_{\alpha} \eta_{\alpha}= - R_0(z) V_{\alpha} \sum_{\beta \neq \alpha} \eta_{\beta}
\ee
Applying the operator $R_{\alpha}(z) (H_0 -z)$ to both sides of the above equation, we obtain the Faddeev equations
\be\label{faddeq} 
\eta_{\alpha}= - R_{\alpha} (z) V_{\alpha} \sum_{\beta \neq \alpha} \eta_{\beta} 
\ee
Thus, we conclude that if $\psi$ is a solution of \eqref{eigH},  then $\psi = \sum_{\alpha} \eta_{\alpha}$ and $\eta_{\alpha}$ are solutions of \eqref{faddeq}. The converse is also true, i.e., if $\eta_{\alpha}$ are solutions of \eqref{faddeq},  then $\psi = \sum_{\alpha} \eta_{\alpha}$ is a solution of \eqref{eigH}. 
As it is well known, a suitable iterated form of Faddeev equations is characterized by a compact operator and this is the main advantage of Faddeev equations with respect to equation~\eqref{LS}.
	
\n
In the above derivation  we have not used the symmetry property of our system and therefore  it is valid for a generic three-particle system. In order to take into account of the fermionic symmetry we proceed as follows.

\n
For notational convenience we describe the symmetry by an operator $T$  on $L^2(\R^6)$. In the coordinates $(k_{23}, p_{23})$, $T$ is defined by
\begin{equation}
	(T\psi) (k_{23},p_{23})=-\psi\left(\frac{k_{23}}{m+1}-\frac{m(m+2)}{(m+1)^2}p_{23},-k_{23}-\frac{p_{23}}{m+1}\right)
\end{equation}
and then  we rewrite
\begin{equation}\label{eq:Hil2}
	\mathcal{H}=\big\{\psi\in L^2(\R^6)\;|\;  T\psi=\psi \big\}.
\end{equation}
In the coordinates $(p_{23},p_{31})$ and $(k_{12},p_{12})$ we have 
\vsa
\begin{equation}
	(T\psi) (p_{23},p_{31})=-\psi(p_{31},p_{23})\,, \;\;\;\;\;\;\;\;\;\; (T\psi) (k_{12},p_{12})=-\psi(-k_{12},p_{12}).
\end{equation}

\n
A direct computation shows that $T$ commutes with $H_0$ and $V_{12}$ and it satisfies $TV_{23}=V_{31}T$ and $TV_{31}=V_{23}T.$  Indeed, recalling the expression of $H_0$ in the coordinates $(p_{23},p_{31})$ given in \eqref{eq:H_0(p)}, we  immediately get $TH_0=H_0T.$  By \eqref{intfour} and assumption \ref{ass2}, one also obtains $TV_{12}=V_{12}T.$ 
Finally, using the equations \eqref{eq:p,q} and  assumption \ref{ass2} we  write
\begin{align*}
	(V_{23}\psi) (p_{23},p_{31})&=\frac{1}{(2\pi)^{3/2})}\int dp\, \hat{v}(p-p_{31})\psi(p_{23},p)\\
	(V_{31}\psi) (p_{23},p_{31})&=\frac{1}{(2\pi)^{3/2}}\int dp\, \hat{v}(p-p_{23})\psi(p,p_{31})
\end{align*}
which  imply $TV_{23}=V_{31}T$ and $TV_{31}=V_{23}T.$  

\n
Using the above properties,  from the definition of $\eta_\alpha$ given in \eqref{eq:defetaa} we obtain $T\eta_{12}=\eta_{12}$, $\eta_{31}=T\eta_{23}$ and the system \eqref{faddeq} reduces to 
\be\label{faddeq2}
	\left\{
	\begin{aligned}
		\eta_{23}&=-R_{23}(z)V_{23}(T\eta_{23}+\eta_{12})\\
		\\
		\eta_{12}&=-R_{12}(z)V_{12}(I+T)\eta_{23}
	\end{aligned}
	\right.
\ee
where $\eta_{23}\in L^2(\R^6)$ and $\eta_{12}\in\mathcal{H}.$ Consequently, in our fermionic system the solution of the eigenvalue equation~\eqref{eigH} reads $\psi= \eta_{12} + (I + T) \,\eta_{23}$. 

\vsa

\n
In this paper we find convenient to use a symmetrized form of the Faddeev equations similar to the one used in 	 \cite{sobolev}. In order to derive such equations, we first introduce  the following bounded and positive operators on $L^2(\R^6)$ for $z<0$
\be
W_\alpha (z)=I+|V_\alpha|^{1/2}R_\alpha(z)|V_\alpha|^{1/2}.
\ee
Using the resolvent identity $R_{\alpha}(z)= R_0(z)- R_0(z) V_{\alpha} R_{\alpha}(z)$,  we find 
\be\label{eq:W}
	W_\alpha(z)=\left(I-|V_{\alpha}|^{1/2}R_0(z)|V_\alpha|^{1/2}\right)^{-1}
\ee

\vsa

\n
Moreover, we define the resolvent of $h_{\alpha}$
\be
r_{\alpha} (z) = (h_{\alpha}-z)^{-1}
\ee

\vsa

\n
which is  a bounded operator in $L^2(\R^3)$  for any  $z<0$, the bounded and positive operator

\be\label{walpha(z)}
	w_{\alpha}(z)=I+|v_{\alpha}|^{1/2}r_{\alpha}(z)|v_{\alpha}|^{1/2}
	\ee

\vsa
\n
and the Fourier transform $\four_\alpha$   with respect to $x_\alpha$
\be
(\four_{\alpha} f) (k_{\alpha}, p_{\alpha}) = \f{1}{(2 \pi)^{3/2}} \int \! d x_{\alpha}  \, e^{-i k_{\alpha}\cdot x_{\alpha}} \, f(x_{\alpha}, p_{\alpha}).
\ee

\n
Then one verifies that
\be\label{eq:Ww}
W_\alpha(z)=\four_\alpha w_\alpha\left(z-\frac{p^2_\alpha}{2n_\alpha}\right)\four_\alpha^*.
\ee

\vsa
\n
Let us  reconsider the first equation in \eqref{faddeq2}.  Using the resolvent identity and taking into account that $V_{\alpha}=-|V_{\alpha}|$, we have 
\begin{equation}\label{eq:eta23}
	\begin{aligned}
	\eta_{23}&=-\big(R_0(z)-R_0(z)V_{23}R_{23}(z) \big)  V_{23}(T\eta_{23}+\eta_{12})\\
					 &=\big( R_0(z)+R_0(z)|V_{23}|R_{23}(z) \big)  |V_{23}|(T\eta_{23}+\eta_{12})\\
					 &=R_0(z)|V_{23}|^{1/2}W_{23}(z)^{1/2}W_{23}(z)^{1/2}|V_{23}|^{1/2}(T\eta_{23}+\eta_{12})
	\end{aligned}
\end{equation}
Analogously, the second equation in \eqref{faddeq2} can be rewritten as
\be\label{eq:eta12}
	\eta_{12}=R_0(z)|V_{12}|^{1/2}W_{12}(z)^{1/2}W_{12}(z)^{1/2}|V_{12}|^{1/2}(I+T)\eta_{23}.
\ee

\vsa
\n
Now we apply $T$ to both sides of \eqref{eq:eta23} and sum it with \eqref{eq:eta12}. On the resulting equation we apply $W_{23}^{1/2}|V_{23}|^{1/2}.$ Denoting
\be
	\begin{aligned}
		\psi_{23}&= W_{23} (z)^{1/2} |V_{23}|^{1/2} (T\eta_{23}+\eta_{12})\\
		\psi_{12}&= W_{12}(z)^{1/2}|V_{12}|^{1/2}(I+T)\eta_{23}
	\end{aligned}
\ee
we find 

\vsa
\begin{equation}
	\psi_{23}=W_{23}(z)^{1/2}|V_{23}|^{1/2}R_0(z) \, \big(T|V_{23}|^{1/2}W_{23}(z)^{1/2}\psi_{23}+|V_{12}|^{1/2}W_{12}(z)^{1/2}\psi_{12} \big).
\end{equation}

\vsa
\n
On the other hand, applying $W_{12}(z)^{1/2}|V_{12}|^{1/2}(I+T)$ to both sides of \eqref{eq:eta23} we find
\begin{equation}
	\psi_{12}=W_{12}(z)^{1/2}|V_{12}|^{1/2}R_0(z)(I+T)|V_{23}|^{1/2}W_{23}(z)^{1/2}\psi_{23}.
\end{equation}
\vsa
\n
Hence we have the following symmetrized form of the Faddeev equations for our model of two identical fermions and a different particle
\be
	\Psi = \boldsymbol{A} (z) \Psi
\ee
where $\Psi = (\psi_{23}, \psi_{12})$ and $\boldsymbol{A}(z)$ is a $2\times 2$ matrix operator  acting on the space
\begin{equation}\label{eq:K}
	\mathcal{K}=L^2(\R^6)\times \mathcal{H}
\end{equation}
 defined by
\begin{equation}\label{eq:def_A}
	\boldsymbol{A}(z)=\boldsymbol{W}(z)^{1/2}\boldsymbol{U}(z)\boldsymbol{W}(z)^{1/2} ,
\end{equation}
with
\be
	\boldsymbol{W}(z)^{1/2}=\diag\{W_{23}(z)^{1/2},W_{12}(z)^{1/2}\}
\ee
and
\be
	\boldsymbol{U}(z)=
	\begin{pmatrix}
		|V_{23}|^{1/2}R_0(z)T|V_{23}|^{1/2} & |V_{23}|^{1/2}R_0(z)|V_{12}|^{1/2}\\
		\\
		|V_{12}|^{1/2}R_0(z)(I+T)|V_{23}|^{1/2} & 0
	\end{pmatrix}.
\ee

\vsa
\n
The advantage of such symmetrized form of the  equations is the fact that  $\boldsymbol{A}(z)$ is compact for $z<0$. 

\begin{theorem}\label{th:A_comp}
	For $z<0$ the operator $\boldsymbol{A}(z)$ is  compact and it is continuous in $z.$
\end{theorem}

\n
The proof of theorem~\ref{th:A_comp} goes exactly as that of  theorem 4.1 in \cite{sobolev} and it is omitted.

\n
It turns out that the number of eigenvalues of $H$ smaller than $z<0$ equals  the number of  eigenvalues of $\boldsymbol{A}(z)$ larger than 1. In order to prove this fact it is useful to introduce the following definition.
\begin{definition}
	Let $B$ be a selfadjoint operator on the Hilbert space $\mathfrak{h}$ and let $\lambda\in\R.$ We set
	\begin{equation}
		n(\lambda,B)=\sup_{\mathfrak{h}_B(\lambda)}\dim \mathfrak{h}_B(\lambda)
	\end{equation}
	where $\mathfrak{h}_B(\lambda)$ denotes any subspace of $\mathcal{D}(B)$ such that if $f$ belongs to $\mathfrak{h}_B(\lambda)$ then $(Bf,f)>\lambda\norma{f}^2.$
\end{definition}

	\n
We stress that if the spectrum of the operator $B$ on the right of $\lambda$ is purely discrete then $n(\lambda,B)$ coincides with the number of eigenvalues (with multiplicities) on the right of $\lambda.$ Thus in particular 
\be
N(z)=n(-z,-H).
\ee

\vsa
\n
Due to the compactness of $\boldsymbol{A}(z)$ stated in theorem~\ref{th:A_comp}, we also have that $n(1,\boldsymbol{A}(z))$ equals the number of eigenvalues of $\boldsymbol{A}(z)$ larger than $1.$

\n


\n
In the next theorem we prove a ``Birman-Schwinger Principle'' for our three-particle system, which is crucial for our analysis.
\begin{theorem}\label{th:fadd} For $z<0$ we have 
	\begin{equation}\label{eq:N(z)}
		N(z)=n(1,\boldsymbol{A}(z)).
	\end{equation}
\end{theorem}

\begin{proof}
	We adapt the  proof of Theorem~3.1  in \cite{sobolev} to our case.  First we show that
		\be\label{n1}
		N(z)=n \big(1,R_0(z)^{1/2} |V| R_0(z)^{1/2} \big)
	\ee
	where $|V|=|V_{23}|+T|V_{23}|+|V_{12}|.$ Indeed, let $x\in\mathfrak{h}_{-H}(-z)$ then
	\begin{equation}
		((H_0-z)x,x)<(|V|x,x).
	\end{equation}
	Setting  $y=(H_0-z)^{1/2}x$ and using selfadjointness of $(H_0-z)^{1/2}$ and $R_0(z)^{1/2}$ we have
	\begin{equation}
		(R_0(z)^{1/2}|V|R_0(z)^{1/2}y,y)>(y,y)
	\end{equation}
	that is $y\in \mathfrak{h}_{R_0(z)^{1/2}|V|R_0(z)^{1/2}}(1).$ This implies $n(1,R_0(z)^{1/2}|V|R_0(z)^{1/2})\geq N(z).$ Reversing the argument we get the opposite inequality.
	\\Next we introduce the matrix operator $\boldsymbol{L}(z)$ on $\Hil^2$ defined by 


\vsa
	\begin{equation}
		\boldsymbol{L}(z)= \f{1}{2} 
		\begin{pmatrix}
			R_0(z)^{1/2}|V| R_0(z)^{1/2}\;&\; R_0(z)^{1/2}|V|R_0(z)^{1/2}\\
			\\
			R_0(z)^{1/2}|V|R_0(z)^{1/2} \;& \;R_0(z)^{1/2}|V|R_0(z)^{1/2}
		\end{pmatrix}
	\end{equation}
	
	\vsa

\vsa
\n
 and we note that
	\be\label{n2}
		n\big(1,R_0(z)^{1/2} |V|R_0(z)^{1/2} \big)=n(1,\boldsymbol{L}(z)).
	\ee
	The operator $\boldsymbol{L}(z)$ can be written as 
	
	\begin{equation}
	\boldsymbol{L}(z)=\boldsymbol{S}(z)\boldsymbol{S}^*(z)
	\end{equation}
	
	\n
	where

	\begin{equation}
		\boldsymbol{S}(z)=
		\begin{pmatrix}
			\frac{1}{2}R_0(z)^{1/2}(I+T)|V_{23}|^{1/2} & \f{1}{\sqrt{2}} R_0(z)^{1/2}|V_{12}|^{1/2}\\
			\\
			\frac{1}{2}R_0(z)^{1/2}(I+T)|V_{23}|^{1/2} & \f{1}{\sqrt{2}}  R_0(z)^{1/2}|V_{12}|^{1/2}
		\end{pmatrix}
	\end{equation}
	
	\vsa
	\n
	is an operator acting from $\mathcal{K}$ to  $\mathcal{H}^2$ and 
	 its adjoint  $\boldsymbol{S}^*(z)$, acting from $\mathcal{H}^2$ to $\mathcal{K}$, is given by 
	\vsa
	\begin{equation}
		\boldsymbol{S}^*(z)=
		\begin{pmatrix}
			|V_{23}|^{1/2}R_0(z)^{1/2} & |V_{23}|^{1/2}R_0(z)^{1/2}\\
			\\
			\f{1}{\sqrt{2}} |V_{12}|^{1/2}R_0(z)^{1/2} & \f{1}{\sqrt{2}} |V_{12}|^{1/2}R_0(z)^{1/2}
		\end{pmatrix}
	\end{equation}
	
	\vsa
	\n
	where we have used $T=I$ on the space $\mathcal{H}$.	Since  $n(\lambda,BB^*)=n(\lambda,B^*B)$ for any bounded operator $B$ (see e.g., lemma 4.2 in \cite{sobolev}), we have  
	
\be\label{n3}
n(1,\boldsymbol{L}(z))=n(1,\boldsymbol{S}^*(z)\boldsymbol{S}(z))
\ee

\vsa
\n
 where $\boldsymbol{S}^*(z)\boldsymbol{S}(z)$ 
  is an operator on $\mathcal{K}$  explicitly given by 
	\begin{equation}
		\boldsymbol{S}^*(z)\boldsymbol{S}(z)=
		\begin{pmatrix}
			|V_{23}|^{1/2}R_0(z)(I+T)|V_{23}|^{1/2} & \sqrt{2}|V_{23}|^{1/2}R_0(z)|V_{12}|^{1/2}\\
			\\
			\frac{\sqrt{2}}{2}|V_{12}|^{1/2}R_0(z)(I+T)|V_{23}|^{1/2} &|V_{12}|^{1/2}R_0(z)|V_{12}|^{1/2}
		\end{pmatrix}.
	\end{equation}

\vsa
\n	
Let us decompose the above operator as follows 		
\be
\boldsymbol{S}^*(z)\boldsymbol{S}(z)=\boldsymbol{D}_1(z)+ \boldsymbol{D}_2(z)
\ee
 where  
 \be
\boldsymbol{D}_1(z)=
\begin{pmatrix}
	|V_{23}|^{1/2}R_0(z)|V_{23}|^{1/2} & 0\\
	\\
	0 &|V_{12}|^{1/2}R_0(z)|V_{12}|^{1/2}
\end{pmatrix}
 \ee

\vsa
\n
and
\vsa
\n
\begin{equation}
	\boldsymbol{D}_2(z)=
	\begin{pmatrix}
		|V_{23}|^{1/2}R_0(z)T|V_{23}|^{1/2} & \sqrt{2}|V_{23}|^{1/2}R_0(z)|V_{12}|^{1/2}\\
		\\
		\frac{\sqrt{2}}{2}|V_{12}|^{1/2}R_0(z)(I+T)|V_{23}|^{1/2} & 0
	\end{pmatrix}.
\end{equation}
 
 \vsa
 \n
 Moreover, let us define 
 \vsa
	\be\label{eq:rel_A_L}
		\tilde{\boldsymbol{A}}(z)
= (I-\boldsymbol{D}_1(z) )^{-1/2}\boldsymbol{D}_2(z)(I-\boldsymbol{D}_1(z))^{-1/2}	
\ee

 \n
and note that, from the definition \eqref{eq:W}, it follows $(I-\boldsymbol{D}_1(z))^{-1/2}=\boldsymbol{W}(z)^{1/2}$.

\n
Let us prove that 
\vsa
	\begin{equation}\label{n4}
		n(1,\boldsymbol{S}^*(z) \boldsymbol{S}(z)): = n(1, \boldsymbol{D}_1(z) + \boldsymbol{D}_2(z) ) = n(1,\tilde{\boldsymbol{A}}(z)).
	\end{equation}
	
	\vsa
	\n
Assume $x\in \mathfrak{h}_{\tilde{\boldsymbol{L}}(z)}(1),$ i.e., 
\begin{equation}
	((\boldsymbol{D}_1(z)+\boldsymbol{D}_2(z))x,x)>(x,x)
\end{equation} 
then
\begin{equation}
	(\boldsymbol{D}_2(z)x,x)>((I-\boldsymbol{D}_1(z))x,x).
\end{equation}

\vsa
\n
Defining $y=(I-\boldsymbol{D}_1(z))^{1/2}x$ and using \eqref{eq:rel_A_L} we get 
\begin{equation}
	(\tilde{\boldsymbol{A}}(z)y,y)>(y,y)
\end{equation}
which means $y\in\mathfrak{h}_{\tilde{\boldsymbol{A}}(z)} (1).$ This proves $n(1,\boldsymbol{S}^*(z) \boldsymbol{S}(z))\leq n(1,\tilde{\boldsymbol{A}}(z)).$ To get the opposite inequality it is sufficient to reverse the argument.

	\n
	We also note that for $z<0$ the operator  $\tilde{\boldsymbol{A}}(z)$ is compact  and it is continuous in $z.$ 
	
	\n
	Finally,  by a direct computation one verifies that 
	
	\begin{equation}
	\tilde{\boldsymbol{A}}(z)=
	\begin{pmatrix}
		1 & 0\\
		0 & \f{\sqrt{2}}{2}
	\end{pmatrix} \boldsymbol{A}(z) 
	\begin{pmatrix}
		1 & 0\\
		0 & \sqrt{2}
	\end{pmatrix}.
\end{equation}
This implies that $ \tilde{\boldsymbol{A}}(z)$ and $\boldsymbol{A}(z)  $ have the same eigenvalues and if $\tilde{\Psi}=(\psi_{23},\psi_{12})$ is an eigenfunction of $\tilde{\boldsymbol{A}}(z)$  then $\Psi=(\psi_{23},\sqrt{2}\psi_{12})$ is an eigenfunction of $\boldsymbol{A}(z)$ with the same eigenvalue.
Thus, in particular
	\be\label{n5}
	n(1,\tilde{\boldsymbol{A}}(z))=n(1,\boldsymbol{A}(z))
	\ee
	
	\n
Taking into account of \eqref{n1}, \eqref{n2}, \eqref{n3}, \eqref{n4}, \eqref{n5}, we conclude the proof.

\end{proof}

\vsa
\n
We conclude this section describing the behavior of the operators $W_{\alpha}(z)$ when $z<0$ is small. 
\\Let us introduce the multiplication operator   in $L^2(\R^6)$ 
\be\label{Gama}
(\Gamma_\alpha (z) f)(k_{\alpha},p_{\alpha}) = \gamma\left(\frac{p_{\alpha}^2}{2n_\alpha}-z\right) f(k_{\alpha},p_{\alpha})
\ee
where $\gamma \in C^{\infty}(\R_+)$ is such that $\gamma (t) >0$ for all $t,$ $\gamma(t)=t$ if $t\leq 1$ and $\gamma(t)=1$ if $t\geq 2.$ 

\n
Moreover, for the resonant pair $23$, we define the operator in $L^2(\R^6)$
\be\label{Pi23}
(\Pi_{23} f) (k_{23}, p_{23}) = \f{1}{\|\varphi\|} \hat{\varphi}  (k_{23}) \int \!\!dk \, f(k,p_{23}) \overline{\hat{ \varphi}(k)}
\ee
where $\varphi$ is the eigenfunction of $|v|^{1/2} g_0 |v|^{1/2}$ with eigenvalue $1$ (see definition~\ref{defre}). 

\n
Using the relation \eqref{eq:Ww} and  Lemma~\ref{lemma:res},  we find 

	\begin{equation}\label{eq:U_res}
		W_{23}(z)^{1/2}=\Gamma_{23}(z)^{-1/4} \, \Pi_{23}+\Gamma_{23}(z)^{-\frac{1-\delta}{4}} \, \tilde{W}_{23}^{(\delta)}(z)
	\end{equation}
	
	\vsa
	\n
	where 
	 $\delta < \min \{1,b-3\}$ and $\tilde{W}_{23}^{(\delta)}(z)$ is continuous in $z\leq0.$
	\\On the other hand, Lemma~\ref{lemma:no_res} implies that  $W_{12}(z)^{1/2}$ is continuous in $z\leq 0.$ 


\vs\vs

\section{\texorpdfstring{Leading term of $\mathbf{A} (z)$ for $z \rightarrow 0^-$}{Leading term of A(z) for z->0-}}\label{sect:lead_term}

The proof of our main result expressed in theorem~\ref{th:main1} requires, via theorem~\ref{th:fadd}, an asymptotic analysis of $n(1, \boldsymbol{A}(z) )$ for $z \rightarrow 0^-$. From theorem~\ref{th:A_comp}, we know that for $z<0$ the operator  $\boldsymbol{A}(z)$ is compact but there is a lack of compactness for $z=0$ and this is the reason why we find that $N(z) $ diverges for $z \rightarrow 0^-$. In this section we shall prove various intermediate results where, at each step, we single out the leading term of $\boldsymbol{A}(z) )$ for $z \rightarrow 0^-$, neglecting operators which are compact for $z\leq 0$. At the end, we shall obtain the following integral operator acting in $L^2((0,R)\times \mathbb{S}^2,dr\otimes d\Omega)$ 
\be\label{SR}
	(S_Rf) (r,\omega)=  \int_0^R \!\!\! d\rho \! \int_{\mathbb{S}^2}\!\! d\Omega(\zeta) \, S(r-\rho, \omega \cdot \zeta) f(\rho, \zeta)  
\ee
where 
\ba\label{SRK}
S(x, y)& \!\!=\!\!& -b(m)  \, \frac{1}{\cosh x+\frac{y}{m+1}} \, , \;\;\;\;\;\;\;\;     
x \in \R,\,\;\;y\in[-1,1]\\
R&\!\!=\!\!&R(z)=\frac{1}{2}|\log|z|| \, ,\\
b(m)&\!\!=\!\!&\frac{1}{4\pi^2}\frac{m+1}{\sqrt{m(m+2)}}\, .
\ea

\n
In  section 6 we shall prove that the asymptotic behavior  of $n(1, \boldsymbol{A}(z) )$ for $z \rightarrow 0^-$ coincides with that of $n(1,S_R)$ for $R \rightarrow +\infty$.

\n
As a first step, we show that the terms in $\boldsymbol{A}(z)$ depending on the interaction between the two fermions give a compact contribution and can be neglected.

	\begin{lemma}\label{laa0}
	For $z\leq0$ the operator  $\boldsymbol{A}(z)-\boldsymbol{A}_0(z)$ is  compact  and it is continuous in $z$, where
	
\be
		\boldsymbol{A}_0(z)=
		\begin{pmatrix}
			A_0(z) & 0\\
			0 & 0
		\end{pmatrix}
	\ee
	and
	\be
		A_0(z)=\Pi_{23}\Gamma_{23}^{-1/4}(z)|V_{23}|^{1/2}R_0(z)T|V_{23}|^{1/2}\Gamma_{23}^{-1/4}(z)\Pi_{23}.
	\ee
\end{lemma}

\begin{proof}
Let us introduce the operators  
\ba
\boldsymbol{\Gamma}(z)&=&\diag\{\Gamma_{23}(z),\Gamma_{12}(z)\}\,, \\
		\tilde{\boldsymbol{W}}^{(\delta)}(z)&=&\diag\{\tilde{W}^{(\delta)}_{23}(z),\tilde{W}^{(\delta)}_{12}(z)\}\,, \\
		\boldsymbol{\Pi} &=&
		\begin{pmatrix}
			\Pi_{23} & 0\\
			0 & 0
		\end{pmatrix}
		\ea
where $\delta < \min \{1,b-3\}$, $\tilde{W}^{(\delta)}_{12}(z)  =\Gamma_{12}(z)^{\frac{1-\delta}{4}}W_{12}(z)$ is continuous in $z \leq 0$ and the other terms have been defined in \eqref{Gama}, \eqref{Pi23}, \eqref{eq:U_res}.

\n
Using the above notation we write
\be
		\boldsymbol{A}(z)=\boldsymbol{A}_0(z)+\boldsymbol{R}(z)
	\ee
	\vsa
	\n
	where
	\ba
		\boldsymbol{R}(z)&=& \boldsymbol{\Pi} \, \boldsymbol{U}^{(1/4,(1-\delta)/4)}(z)\tilde{\boldsymbol{W}}^{(\delta)}(z)+\tilde{\boldsymbol{W}}^{(\delta)}(z)\boldsymbol{U}^{((1-\delta)/4,(1-\delta)/4)}(z)\tilde{\boldsymbol{W}}^{(\delta)}(z)\nonumber \\
&+&\tilde{\boldsymbol{W}}^{(\delta)}(z)\boldsymbol{U}^{((1-\delta)/4,1/4)}(z)\boldsymbol{\Pi}
	\ea
and
\be
		\boldsymbol{U}^{(\mu,\nu)}(z)=\boldsymbol{\Gamma}(z)^{-\mu}\boldsymbol{U}(z)\boldsymbol{\Gamma}(z)^{-\nu}, \;\;\;\;\;\;\;\; 0\leq \mu, \nu \leq \f{1}{4}, \;\;\; \mu + \nu <\f{1}{2}\,.
	\ee
	By lemma 4.4 in \cite{sobolev}, for $z\leq 0$ the operator $ \boldsymbol{U}^{(\mu,\nu)}(z)$ is compact and it is continuous in $z$. Since $\tilde{W}_{\alpha}^{(\delta)}(z)$, $z \leq 0$,  and $\Pi_{23}$ are bounded, we conclude that for $z\leq 0$ the operator $\boldsymbol{R}(z)$ is compact and it is continuous in $z$.

\end{proof}

\n
In the next step we reduce the problem to the analysis of an operator in $L^2(\R^3)$. Such an operator is   better analysed using the coordinates $(p_{23},p_{31}).$

\begin{lemma}\label{la0b}
For $\lambda>0$ and $z<0$ we have 
	\begin{equation}\label{eq:B}
		n(\lambda,\boldsymbol{A}_0(z))=n(\lambda,\mathcal{B}(z))
	\end{equation}
	where $\mathcal{B}(z)$ is the integral operator in $L^2(\R^3)$ with  kernel
	\begin{equation}\label{eq:kernel_B}
		\mathcal{B}(p,q;z)=-\frac{\widehat{|v|^{1/2}\varphi}\left(p+\frac{q}{m+1}\right)\overline{\widehat{|v|^{1/2}\varphi}\left(q+\frac{p}{m+1}\right)}}{\gamma\left(\frac{p^2}{2n}-z\right)^{1/4}(H_0-z)\,\gamma\left(\frac{q^2}{2n}-z\right)^{1/4}}.
	\end{equation}
\end{lemma}

\begin{proof}
We first observe that  $n(\lambda,\boldsymbol{A}_0)=n(\lambda,A_0).$ Moreover,   $A_0$ is compact for $z< 0$ and then  $n(\lambda,A_0(z))$ is the number of its eigenvalues larger than $\lambda.$ For an eigenvalue $\tilde{\lambda}>\lambda>0$, the eigenvalue equation for $A_0(z)$ is explicitly given by 
	\begin{multline}\label{eaa0}
		\tilde{\lambda}\psi(p_1,p_2)=-\frac{1}{(2\pi)^3}\int dp_2'\,\frac{\widehat{|v|^{1/2}}(p_2-p_2')}{\gamma\left(\frac{p_1^2}{2n}-z\right)^{1/4}\left(H_0(p_1,p_2')-z\right)\gamma\left(\frac{p_2'^2}{2n}-z\right)^{1/4}}\\
		\times\int dp_1'\, \widehat{|v|^{1/2}}(p_1-p_1')\hat{\varphi}\left(p_1'+\frac{p_2'}{m+1}\right)\int dq\, \overline{\hat{\varphi}\left(q+\frac{p_2'}{m+1}\right)}\psi(p_2',q).
	\end{multline}	
	\n
Since
	\be
		\begin{aligned}
			\int dq\,\widehat{|v|^{1/2}}(p_1-q)\hat{\varphi}\left(q+\frac{p_2}{m+1}\right)&=\int dq'\, \widehat{|v|^{1/2}}\left(q'-p_1-\frac{p_2}{m+1}\right)\hat{\varphi}(q')\\
																																								&=(2\pi)^{3/2}\widehat{|v|^{1/2}\varphi}\left(p_1+\frac{p_2}{m+1}\right)
		\end{aligned}
	\ee
	we  rewrite equation~\eqref{eaa0} as follows
	\begin{multline}\label{eaa02}
		\tilde{\lambda}\psi(p_1,p_2)=-\frac{1}{(2\pi)^{3/2}}\int dp_2'\, \frac{\widehat{|v|^{1/2}}(p_2-p_2')\widehat{|v|^{1/2}\varphi}\left(p_1+\frac{p_2'}{m+1}\right)}{\gamma\left(\frac{p_1^2}{2n}-z\right)^{1/4}(H_0(p_1,p_2')-z)\gamma\left(\frac{p_1^2}{2n}-z\right)^{1/4}}
		\\ \times\int dq\,\overline{\hat{\varphi}\left(q+\frac{p_2'}{m+1}\right)}\psi(p_2',q).
	\end{multline}
	Let us define
	\be
		\xi(p)=\int dq\,\overline{\hat{\varphi}\left(q+\frac{p}{m+1}\right)}\psi(p,q).
	\ee
	Then $\xi\in L^2(\R^3)$ and, by \eqref{eaa02}, it satisfies the equation
	\be
		\begin{aligned}
			\tilde{\lambda}\xi(p_1)&=-\frac{1}{(2\pi)^{3/2}}\int dp_2\, \overline{\hat{\varphi}\left(p_2+\frac{p_1}{m+1}\right)}\\
														 &\phantom{{=}}\times \int dp_2'\frac{\widehat{|v|^{1/2}}(p_2-p_2')\widehat{|v|^{1/2}\varphi}\left(p_1+\frac{p_2'}{m+1}\right)}{\gamma\left(\frac{p_1^2}{2n}-z\right)^{1/4}(H_0(p_1,p_2')-z)\gamma\left(\frac{p_2'^2}{2n}-z\right)^{1/4}}\xi(p_2')\\
														 &=-\int dq\, \frac{\widehat{|v|^{1/2}\varphi}\left(p_1+\frac{q}{m+1}\right)\overline{\widehat{|v|^{1/2}\varphi}\left(q+\frac{p_1}{m+1}\right)}}{\gamma\left(\frac{p_1^2}{2n}-z\right)^{1/4}(H_0(p_1,q)-z)\gamma\left(\frac{q^2}{2n}-z\right)^{1/4}}\xi(q)
		\end{aligned}
	\ee
	that is $\mathcal{B}(z)\xi=\tilde{\lambda}\xi.$  
	On the other hand, if $\xi \in L^2(\R^3)$ is such that $\mathcal{B}(z)\xi=\tilde{\lambda}\xi$ then 
	\be 
		\psi(p_1,p_2)=-\frac{1}{(2\pi)^{3/2}}\int dq\, \frac{\widehat{|v|^ {1/2}}(p_2-q)\widehat{|v|^{1/2}\varphi}\left(p_1+\frac{q}{m+1}\right)}{\gamma\left(\frac{p_1^2}{2n}-z\right)^{1/4}(H_0(p_1,q)-z)\gamma\left(\frac{q^2}{2n}-z\right)^{1/4}}\, \xi(q)
	\ee
	satisfies the equation  $A_0(z)\psi=\tilde{\lambda}\psi$ and therefore the lemma is proved.

\end{proof}

\n
The lack of compactness of $\mathcal{B}(z)$ for $z=0$ is clearly due to the behavior of its integral kernel near the origin. Indeed, in the following Lemma we show that the  difference of $\mathcal{B}(z)$ with an operator whose kernel is different from zero only in a ball of radius one  is compact and continuous in $z\leq0.$ Denoted by $\chi_a$ the characteristic function of the ball of radius $a>0$, we have 

\begin{lemma}\label{lB}
For $z \leq 0$ the operator $\mathcal{B}(z)-\tilde{\mathcal{B}}(z)$ is compact and it is continuous in $z$, where $\tilde{\mathcal{B}}(z)$ is the operator in $L^2(\R^3)$ with  kernel
	\be
		\tilde{\mathcal{B}}(p,q;z)=-\,\frac{1}{c(m)}\frac{\chi_1 (p) \chi_1 (q)}{\left(\frac{p^2}{2n}-z\right)^{1/4}(H_0(p,q)-z)\left(\frac{q^2}{2n}-z\right)^{1/4}}
	\ee
 and $c(m)=2^{5/2}\pi^2\left(\frac{m}{m+1}\right)^{3/2}$.	
Moreover
\be\label{boh}
n(\lambda,\tilde{\mathcal{B}}(z))=n(\lambda,\mathcal{B}_0(z))
\ee
where $\mathcal{B}_0(z)$ is the integral operator in $L^2(\R^3)$  with  kernel 
	\be
		\mathcal{B}_0(p,q;z)=-\frac{1}{c(m)}\frac{\chi_{|z|^{-1/2}}(p)  \chi_{|z|^{-1/2}}(q)}{\left(\frac{p^2}{2n}+1\right)^{1/4}(H_0(p,q)+1)\left(\frac{q^2}{2n}+1\right)^{1/4}}.
	\ee
\end{lemma}

\begin{proof}
	The proof  is divided into three steps. We first introduce the operator $\mathcal{E}_0(z)$ with integral kernel
	\be
		\mathcal{E}_0(p,q;z)=-\chi_1(p)\frac{\widehat{|v|^{1/2}\varphi}\left(p+\frac{q}{m+1}\right)\overline{\widehat{|v|^{1/2}\varphi}\left(q+\frac{p}{m+1}\right)}}{\gamma\left(\frac{p^2}{2n}-z\right)^{1/4}(H_0(p,q)-z)\gamma\left(\frac{q^2}{2n}-z\right)^{1/4}}\chi_1(q)
	\ee
and note that the integral kernel of the difference $\mathcal{B}(z)-\mathcal{E}_0(z)$  can be written as 
	\begin{align}
		(\textup{I}+\textup{II}+\textup{III})(p,q;z)\!=&(1-\chi_1(p))\frac{\widehat{|v|^{1/2}\varphi}\left(p+\frac{q}{m+1}\right)\overline{\widehat{|v|^{1/2}\varphi}\left(q+\frac{p}{m+1}\right)}}{\gamma\left(\frac{p^2}{2n}-z\right)^{1/4}(H_0(p,q)-z)\gamma\left(\frac{q^2}{2n}-z\right)^{1/4}}\chi_1(q)\nonumber \\
														&+\chi_1(p)\frac{\widehat{|v|^{1/2}\varphi}\left(p+\frac{q}{m+1}\right)\overline{\widehat{|v|^{1/2}\varphi}\left(q+\frac{p}{m+1}\right)}}{\gamma\left(\frac{p^2}{2n}-z\right)^{1/4}(H_0(p,q)-z)\gamma\left(\frac{q^2}{2n}-z\right)^{1/4}}(1-\chi_1(q))\nonumber \\
														&+(1-\chi_1(p))\frac{\widehat{|v|^{1/2}\varphi}\left(p+\frac{q}{m+1}\right)\overline{\widehat{|v|^{1/2}\varphi}\left(q+\frac{p}{m+1}\right)}}{\gamma\left(\frac{p^2}{2n}-z\right)^{1/4}(H_0(p,q)-z)\gamma\left(\frac{q^2}{2n}-z\right)^{1/4}}(1-\chi_1(q)).
	\end{align}
	Let us consider $\textup{I}(p,q;z).$ We note that $|v|^{1/2}\varphi\in L^1(\R^3)$ which implies $\widehat{|v|^{1/2}\varphi}\leq c.$ Moreover, $H_0(p,q)-z\geq c \, p^2$ and therefore we obtain
	\be
		\textup{I}(p,q;z)\leq c\, \frac{(1-\chi_1(p))\,\chi_1(q)}{p^2\,q^{1/2}}
	\ee
	which is a square integrable function. Analogously, one can prove square integrability of $\textup{II}(p,q;z).$ 
	In order to estimate $\textup{III}(p,q;z)$, we note that $\widehat{|v|^{1/2}\varphi}\in L^2(\R^3)$ and therefore
	\be
		\textup{III}(p,q;z)\leq c(1-\chi_1(p))\frac{\widehat{|v|^{1/2}\varphi}\left(p+\frac{q}{m+1}\right)}{q^2}(1-\chi_1(q))
	\ee
	is square integrable. 
	Hence we conclude that for $z \leq 0$ the operator $\mathcal{B}(z)-\mathcal{E}_0(z)$ is  Hilbert-Schmidt  and it is continuous in $z.$ 
	\n
	Now we consider the operator $\mathcal{E}_1(z)$ with integral kernel
	\be
		\mathcal{E}_1(p,q;z)=-\frac{1}{c(m)}\, \frac{\chi_1(p) \chi_1(q)}{\gamma\left(\frac{p^2}{2n}-z\right)^{1/4} \!\! (H_0(p,q)-z) \, \gamma\left(\frac{q^2}{2n}-z\right)^{1/4}}\, .
	\ee
We note that 
	\be\label{pif}
		|\widehat{|v|^{1/2}\varphi}(k)-\widehat{|v|^{1/2}\varphi}(0)|\leq c \, |k|^\nu \,, \qquad 0<\nu<\frac{b-3}{2}.	
	\ee
	Indeed, using  $|e^{-ik\cdot x}-1|\leq c\,  |k|^\nu |x|^\nu$, we have  
	\begin{align*}
		|\widehat{|v|^{1/2}\varphi}(k)-\widehat{|v|^{1/2}\varphi}(0)|&\leq c\,  |k|^\nu \int dx\, |x|^\nu |v|^{1/2}(x)|\varphi|(x)\\
																																 &\leq c \, \|\varphi\|\,|k|^\nu \left(\int dx\, |x|^{2\nu } |v|(x)\right)^{1/2}
	\end{align*}
	and the last integral is finite  by assumption  \ref{ass1}. Moreover,  using Young's inequality, we get
	\be\label{kk'}
		H_0(p,q)-z=\frac{p^2}{2\mu}+\frac{q^2}{2\mu}+\frac{p\cdot q}{m}\geq c\left(\frac{(p^{2\kappa})^{1/\kappa}}{\kappa}+\frac{(q^{2\kappa'})^{1/\kappa'}}{\kappa'}\right)\geq c\, p^{2\kappa}q^{2\kappa'}
	\ee

\n
for any $\kappa,\kappa'>0$ such that $\kappa+\kappa'=1.$
	
	\vsa
	\n
By \eqref{pif}, \eqref{kk'} and condition \eqref{eq:norm},  we find 
	\begin{align}
		|\mathcal{E}_0(p,q;z)-\mathcal{E}_1(p,q;z)|&\leq c\,\chi_1 (p)\frac{|p|^\nu+|q|^\nu}{\left(\frac{p^2}{2n}-z\right)^{1/4}(H_0(p,q)-z)\left(\frac{q^2}{2n}-z\right)^{1/4}}\chi_1(q)\\
																			 &\leq c\, \f{ \chi_1 (p) \chi_1(q)}{  |p|^{- \nu+2\kappa+ 1/2}|q|^{2\kappa' +1/2} }  
			+ c\, \f{ \chi_1 (p) \chi_1(q)}{  |q|^{- \nu+2\kappa+ 1/2}|p|^{2\kappa' +1/2} }		\end{align}
	which is square integrable choosing $\kappa\in (\frac{1}{2},\frac{1+\delta}{2}).$
	\vsa
	\n
	In order to obtain the operator $\tilde{\mathcal{B}}(z)$ from $\mathcal E_1(z)$, it remains to replace $\gamma\left(\frac{p^2}{2n}-z\right)$ and $\gamma\left(\frac{q^2}{2n}-z\right)$ with $\frac{p^2}{2n}-z$ and $\frac{q^2}{2n}-z$, respectively. One can easily see that the difference $\mathcal{E}_1(z)-\tilde{\mathcal{B}}(z)$ is compact up to $z=0.$ 
	This concludes the proof that  for $z\leq 0$ the operator 
	$\mathcal{B}(z)-\tilde{\mathcal{B}}(z)$ 
	is compact  and it is continuous in $z.$
	
	\n
	In order to prove \eqref{boh} it is sufficient to observe that 
	 $\tilde{\mathcal{B}}(z)$ is unitarily equivalent to $\mathcal{B}_0(z)$ via the unitary operator $U_z$ defined by $U_z \xi(p)=|z|^{3/4}\xi(|z|^{1/2}p).$

\end{proof}

\n
In the following Lemma we finally arrive at the operator  $S_R$ defined in \eqref{SR}, \eqref{SRK}.
\begin{lemma}\label{lS_R}
For $z \leq 0$ the operator $\mathcal{B}_0(z)-\mathcal{S}(z)$ is compact and it is continuous in $z$, where  $\mathcal{S}(z)$ is the integral operator in $L^2(\R^3)$ with  kernel
	\be
	\mathcal{S}(p,q;z)=-\frac{(2n)^{1/2}}{c(m)} \frac{\big( \chi_{|z|^{-1/2}} - \chi_1 \big)  (p)  \big( \chi_{|z|^{-1/2}} - \chi_1 \big)  (q)}{|p|^{1/2}\left(\frac{p^2}{2\mu}+\frac{q^2}{2\mu}+\frac{p\cdot q}{m}\right)|q|^{1/2}} \, .
	\ee	
Moreover
	\be
	n(\lambda,\mathcal{S}(z))=n(\lambda, S_R).
	\ee
\end{lemma}

\begin{proof}
	The first point is easy to check. Let us prove the second statement. 
	Using the unitary operator $M: L^2(\R^3)\to L^2(\R_+\times \mathbb{S}^2,dr\otimes d\Omega)$ defined by 
	
	\be\label{eq:M}
(M\xi) (r,\omega)=e^{3r/2 }\xi(e^r, \omega)
	\ee
 we see that $\mathcal{S}(z)$ is unitarily equivalent to the operator on $L^2(\R_+\times \mathbb{S}^2,dr\otimes d\Omega)$ with integral kernel
\begin{equation}\label{eq:S_R1}
	-b(m) \frac{ \chi_{(0,R)}(x) \chi_{(0,R)}(x')}{\cosh(x-x')+\frac{\zeta\cdot \omega}{m+1}}\, .
\end{equation}
where $\chi_{(a,b)}$ is the characteristic function of the interval $(a,b)$. 
Indeed,  
\be
	\begin{aligned}
		(M\mathcal{S}(z)\xi) (r,\omega)&=-\frac{(m+1)^2}{4\pi^2\sqrt{m^3(m+2)}}e^{3r/2}\int_0^{+\infty}d\rho\, \rho^{3/2}\\
																&\phantom{={}}\times\int_{\mathbb{S}^2}d\Omega(\zeta) \,   \frac{\left(\chi_{|z|^{-1/2}}(e^r,\omega)-\chi_1(e^r,\omega)\right)\left(\chi_{|z|^{-1/2}}(\rho, \zeta)-\chi_1(\rho, \zeta)\right)}{e^{r/2}\left(\frac{e^{2r}}{2\mu}+\frac{\rho^2}{2\mu}+e^r\rho\frac{\omega\cdot\zeta}{m}\right)}\, \xi(\rho,\zeta)\\
															  &=-\frac{(m+1)^2}{4\pi^2\sqrt{m^3(m+2)}}e^r\int_{-\infty}^{+\infty}dx  \, e^{5x /2}\int_{\mathbb{S}^2}d\Omega (\zeta)\, \frac{\chi_{(0,R)}(r)\chi_{(0,R)}(x )}{\frac{e^{2r}}{2\mu}+\frac{e^{2 x }}{2\mu}+\frac{\omega\cdot\zeta}{m} \,e^{r+ x }  } \, \xi(e^{x}, \zeta) \\
																&=-\frac{(m+1)^2}{4\pi^2\sqrt{m^3(m+2)}}e^r \!\!\! \int_{-\infty}^{+\infty}\!\!\!\! d x \,e^{ x }\!\!\!\int_{\mathbb{S}^2}\!\!\!\!d\Omega(\zeta)\, \frac{\chi_{(0,R)}(r)\chi_{(0,R)}( x )}{\frac{e^{r+ x }}{\mu}\left(\frac{e^{r- x }+e^{ x -r}}{2}+\frac{\omega\cdot\zeta}{m+1}\right)}\, e^{3 x/ 2  }\xi(e^{ x }, \zeta)\\
																&=-b(m)\int_{-\infty}^{+\infty}d x \, \int_{\mathbb{S}^2}d\Omega(\zeta)\, \frac{\chi_{(0,R)}(r)\chi_{(0,R)}( x )}{\cosh(r- x )+\frac{\omega\cdot\zeta}{m+1}} (M\xi) (x ,\zeta).
	\end{aligned}
\ee
Since the operator with the kernel given by \eqref{eq:S_R1} can be considered as an operator on $L^2((0,R)\times \mathbb{S}^2,d\rho\otimes d\Omega)$,  we find $S_R.$

\end{proof}

\vs\vs
\section{Proof of theorem~\ref{th:main1}}\label{sei}

In this section we give the proof of theorem~\ref{th:main1}.  Taking into account of theorem~\ref{th:fadd}, the result for $m<m_*$ is obtained  in two steps.  We first show that 
\be\label{SRR}
\f{n(1,S_R)}{2R}
\ee
converges for $R\rightarrow \infty$ (see proposition~\ref{pro1} below). Then  we prove that 
\be\label{Azz}
\f{n(1,\boldsymbol{A} (z) )}{|\log |z||}
\ee
converges to the same limit for $z \rightarrow 0^-$  (see proposition~\ref{pro2} below). 

\n
In order to study the asymptotic behavior of \eqref{SRR} for $R \rightarrow \infty$, it is convenient to decompose the integral kernel of $S_R$ in spherical harmonics. Indeed, denoted by $P_l$  the  Legendre polynomial of order $l$  and by $Y^{\nu}_l$ the spherical harmonic of order $l,\nu$,  we write

\be
S(x,\omega\cdot\zeta)=\sum_{l=0}^{+\infty}\frac{2l+1}{2}P_l(\omega\cdot\zeta)\int_{-1}^1 \!\!dy\, P_l(y)S(x,y)
\ee
 and,  using the addition formula
\begin{equation}
	P_l(\omega\cdot\zeta)=\frac{4\pi}{2l+1}\sum_{\nu=-l}^l\overline{Y_l^{\nu}(\zeta)}Y_l^{\nu}(\omega)\,, 
\end{equation}
we find (see \eqref{SR}, \eqref{SRK})
\be\label{eq:S^R_l}
	\begin{aligned}
		(S_Rf) (r,\omega)&=2\pi\sum_{l=0}^{+\infty}\sum_{\nu=-l}^lY_l^{\nu}(\omega) \! \!\int_0^R \!\!\! d\rho\!\! \int_{-1}^1 \!\!\!dy\, S(r-\rho,y)P_l(y) \! \int_{\mathbb{S}^2} \!\! d\Omega (\zeta) \,f(\rho,\zeta)\overline{Y_l^{\nu}(\zeta)}\\
									&=-2\pi b(m)\sum_{l=0}^{+\infty}\sum_{\nu=-l}^lY_l^{\nu}(\omega) \!  \int_0^R \!\!\!d\rho\,f_{l\nu}(\rho)\int_{-1}^1 \!\!dy\,\frac{P_l(y)}{\cosh(r-\rho)+\frac{y}{m+1}}\\
									&: =\sum_{l=0}^{+\infty}\sum_{\nu=-l}^l  (S_R^{(l)}f_{l\nu}) (\rho)Y_l^{\nu}(\omega)
	\end{aligned}
\ee
where 
\be\label{eq:f_lm}
f_{l\nu} (\rho)= \int_{\mathbb{S}^2} \!\! d\Omega (\zeta) \,f(\rho,\zeta)\overline{Y_l^{\nu}(\zeta)}
\ee
and $S_R^{(l)}$ is the integral operator in $L^2((0,R))$ with  kernel  defined by
\begin{equation}\label{eq:S^l}
	S^{(l)}(x-x')=-2\pi b(m)\int_{-1}^1 dy\,\frac{P_l(y)}{\cosh (x-x')+\frac{y}{m+1}}\,, \;\;\;\;\;\;\;\; x, x' \in \R\, .
\end{equation}
In particular, this decomposition implies 
\be\label{nnl}
	n(\lambda,S_R)=\sum_{l=0}^{+\infty}(2l+1)n(\lambda,S_R^{(l)}).
\ee

\n
We are now in position to characterize the asymptotics of \eqref{SRR}.

\begin{proposition}\label{pro1}
For any $\lambda >0$ we have 
\be\label{nsr}
\lim_{R\to+\infty}\frac{n(\lambda ,S_R)}{2R}=\sum_{l=0}^{+\infty}\frac{2l+1}{4\pi}\left|\left\{k\in \R \, |\, \,\hat{S}^{(l)}(k)> \frac{\lambda}{\sqrt{2\pi}}\right\}\right| 
	\ee
and the limit is continuous in $\lambda >0$.
\end{proposition}  

\begin{proof}

The kernel of the operator $S^{(l)}_R$ is an even function and satisfies the estimate
\be\label{stsl}
|S^{(l)}(x)| \leq \f{c}{\cosh x - \f{1}{m+1}}\,.
\ee
By \eqref{stsl} we have  $ S^{(l)} \in L^1(\R) \cap L^{\infty}(\R)$ and  $e^{\varepsilon |x|} S^{(l)} \in L^2(\R)$ for any $\varepsilon \in [0,1)$. These properties imply that $|\hat{S}^{(l)} (k) | \rightarrow 0$ for $|k| \rightarrow \infty$, by Riemann-Lebesgue theorem, and that $\hat{S}^{(l)} (k)$ has an analytic continuation to a neighborhood of the real axis (see, e.g., theorem IX.13 of \cite{rs2}).  Then, for any $\lambda >0$ the set $\mathfrak{M}(\lambda)=\left\{k\in \R\,|\, \hat{S}^{(l)}(k)=\frac{\lambda}{\sqrt{2\pi}}\right\}$ consists of a finite number of points and, in particular, it has measure zero. Thus, the hypotheses of lemma 4.6 in \cite{sobolev} are satisfied and we have
\be
\lim_{R\rightarrow +\infty} \f{n(\lambda, S^{(l)}_R)}{R} = \frac{1}{2\pi}\int_{-\infty}^{+\infty}\!\!\! dk\,\chi_{\left(\frac{\lambda}{\sqrt{2\pi}},+\infty\right)}(\hat{S}^{(l)}(k))=\frac{1}{2\pi}\left|\left\{k\in\R\,|\,\hat{S}^{(l)}(k)> \frac{\lambda}{\sqrt{2\pi}}\right\}\right|
\ee
where $|A|$ denotes the Lebesgue measure of the set $A$. 
Taking into account of \eqref{nnl}, we obtain \eqref{nsr}. Concerning the continuity of the limit, we observe that $\lim_{\lambda'\to\lambda}\chi_{(\lambda',+\infty)}(\hat{S}^{(l)}(k))=\chi_{(\lambda,+\infty)}(\hat{S}^{(l)}(k))$ for any $k$ such that $\hat{S}^{(l)} (k) \neq \lambda$.  Using the dominated convergence theorem we conclude the proof.

\end{proof}

\n
We collect here some properties of $\hat{S}^{(l)} (k)$ which will be useful in the sequel. By definition, we have

\begin{equation}\label{eq:S^l_hat}
	\begin{aligned}
		\sqrt{2\pi}\hat{S}^{(l)}(k)&=-2\pi\, b(m)\! \int_{-\infty}^{+\infty} \!\!\!dx\, e^{-ikx}\!\int_{-1}^1 \!\! dy\, \frac{P_l(y)}{\cosh{x}+\frac{y}{m+1}}\\
											&=-\frac{1}{2\pi}\frac{m+1}{\sqrt{m(m+2)}}\int_{-1}^1 \!\! dy\, P_l(y)\! \int_{-\infty}^{+\infty}\!\!\! dx\,\frac{e^{-ikx}}{\cosh x+\frac{y}{m+1}}\\
											&=-\frac{1}{\pi}\frac{m+1}{\sqrt{m(m+2)}}\int_{-1}^1 \!\! dy\, P_l(y) \int_0^{+\infty} \!\!\! dx\,\frac{1}{\cosh x+\frac{y}{m+1}}\cos(kx)\\
											&=-\frac{m+1}{\sqrt{m(m+2)}}\int_{-1}^1 \!\!  dy\, P_l(y) \frac{\sinh \left(k\arccos\left(\frac{y}{m+1}\right)\right)}{\sinh (k\pi)\sin\left(\arccos \left(\frac{y}{m+1}\right)\right)}
	\end{aligned}
\end{equation}

\n
where, for the computation of the last integral, we refer the reader to \cite[p. 30]{erdelyi}. Moreover, by the elementary relations $\;\arccos(\alpha)=\frac{\pi}{2}-\arcsin(\alpha)$ and $\sinh(\alpha\pm\beta)=\sinh(\alpha)\cosh(\beta)\pm\sinh(\beta)\cosh(\alpha)$, we find
\begin{equation}
	\begin{aligned}
		\hat{S}^{(l)}(k)=&-\,\frac{1}{\sqrt{2\pi}}\frac{m+1}{\sqrt{m(m+2)}}\int_{-1}^1\!\!dy\, P_l(y)\left[\frac{\cosh\left(k\arcsin\left(\frac{y}{m+1}\right)\right)}{2\cosh\left(k\frac{\pi}{2}\right)\cos\left(\arcsin\left(\frac{y}{m+1}\right)\right)}\right.\\
										 &\left.-\,\frac{\sinh\left(k\arcsin\left(\frac{y}{m+1}\right)\right)}{2\sinh\left(k\frac{\pi}{2}\right)\cos\left(\arcsin\left(\frac{y}{m+1}\right)\right)}\right]
	\end{aligned}
\end{equation}

\vs

\n
and using the parity of the Legendre polynomials in the equation above we  obtain

\begin{equation}\label{eq:S^l_e_o}
		\hat{S}^{(l)}(k)=
		\begin{cases}
			\displaystyle{\frac{1}{\sqrt{2\pi}}\frac{m+1}{\sqrt{m(m+2)}}\int_{0}^1dy\, P_l(y)\frac{\sinh\left(k\arcsin\left(\frac{y}{m+1}\right)\right)}{\sinh\left(k\frac{\pi}{2}\right)\cos\left(\arcsin\left(\frac{y}{m+1}\right)\right)}}\qquad\text{$l$ odd}\\
			{}\\
			-\displaystyle{\frac{1}{\sqrt{2\pi}}\frac{m+1}{\sqrt{m(m+2)}}\int_{0}^1dy\, P_l(y)\frac{\cosh\left(k\arcsin\left(\frac{y}{m+1}\right)\right)}{\cosh\left(k\frac{\pi}{2}\right)\cos\left(\arcsin\left(\frac{y}{m+1}\right)\right)}}\qquad\text{$l$ even}.
		\end{cases}
	\end{equation}
	
\n
Note that $\hat{S}^{(l)}(k)=\hat{S}^{(l)}(-k)$. Moreover, in \cite{CDFMT}
the following properties of $\hat{S}^{(l)}(k)$ are proved

			\begin{equation}\label{eq:S^l_sign}
		\begin{aligned}
			\hat{S}^{(l)}(k)&\geq 0\qquad\text{$l$ odd}\\
			\hat{S}^{(l)}(k)&\leq 0\qquad\text{$l$ even},
		\end{aligned}
	\end{equation}
\vsa
	\begin{equation}\label{eq:S^l_mon}
		\begin{aligned}
			\hat{S}^{(l+2)}(k)&\leq \hat{S}^{(l)}(k) \qquad\text{$l$ odd}\\
			\hat{S}^{(l+2)}(k)&\geq \hat{S}^{(l)}(k) \qquad\text{$l$ even}
		\end{aligned}
	\end{equation}

\n
and
	\be\label{maxs}
	\max _{k\in \R}\hat{S}^{(1)}(k) =  \hat{S}^{(1)}(0)=\frac{1}{\sqrt{2\pi}}\Lambda(m)
	\ee
	 where $\Lambda(m)$ is defined in \eqref{eq:Lambda}.

\n
For the proof of the last step   we make repeatedly use of the following technical lemma (for the proof see lemma 4.9 in \cite{sobolev}).

\begin{lemma}\label{compact}
	Let $B(z)=B_0(z)+K(z)$,  where for $z<0$ $(z \leq 0)$ the operator $B_0(z) $ $(K(z))$  is  compact and  continuous in $z$.  Suppose that for a function $f$ such that $f(z)\to0$ when $z\to0^-$ there exists the limit 
	\begin{equation}
		\lim_{z\to 0^-}f(z)n(\lambda,B_0(z))=l(\lambda)
	\end{equation}
	and $l(\lambda)$ is continuous in $\lambda>0$. Then the following holds
	\begin{equation}
		\lim_{z\to 0^-}f(z)n(\lambda,B(z))=l(\lambda).
	\end{equation}
\end{lemma}

\n
Then we have

\begin{proposition}\label{pro2}
\be\label{nsa}
		\lim_{R\to+\infty}\frac{n(1 ,S_R)}{2R}
= 	\lim_{z\to 0^-}\frac{n(1,\boldsymbol{A}(z))}{|\log|z||}	\,.		
	\ee

\end{proposition}  

\begin{proof}
The proof is obtained using proposition \ref{pro1}, lemma \ref{compact} and lemmas \ref{laa0}, \ref{la0b}, \ref{lB}, \ref{lS_R}.

\end{proof}

\n
Let us prove our main result in the case $m<m_*$.

\vs
\begin{proof}[Proof of theorem~\ref{th:main1} (case $m<m_*$)] 
	By  theorem~\ref{th:fadd} and propositions~\ref{pro1},  \ref{pro2} we find that the limit relation \eqref{eq:main} holds with 	
\begin{equation}\label{eq:C(m)2}
	\mathcal{C}(m):=\sum_{\substack{l=1,\\ l\textup{ odd}}}^{+\infty}\frac{2l+1}{2\pi}\left|\left\{k\in[0,+\infty)\,|\,\hat{S}^{(l)}(k)> \frac{1}{\sqrt{2\pi}}\right\}\right|
\end{equation}
	where we have used the parity of $\hat{S}^{(l)}(k)$ and the sum is only for $l$ odd due to the property \eqref{eq:S^l_sign}. It remains to show that $\mathcal C(m)$ is finite and strictly positive. 
	
\n
Let us prove the upper bound for  $\mathcal C(m)$.  We first  look for an estimate of  $\hat{S}^{(l)}(k)$ for $l$ odd and $k\geq 0.$  
By Cauchy-Schwarz inequality we have
\begin{equation}
	\begin{aligned}
		\hat{S}^{(l)}(k)&\leq \frac{m+1}{\sqrt{2\pi}\sqrt{m(m+2)}}\int_0^1dy\,|P_l(y)|\frac{\sinh\left(k\arcsin\left(\frac{y}{m+1}\right)\right)}{\sinh\left(k\frac{\pi}{2}\right)\cos\left(\arcsin\left(\frac{y}{m+1}\right)\right)}\\
											&\leq \frac{(m+1)^{3/2}}{\sqrt{2\pi}\sqrt{m(m+2)}}\frac{1}{\sqrt{2l+1}}\left[\int_0^{z_0} dz \frac{\sinh^2(kz)}{\sinh^2\left(k\frac{\pi}{2}\right)\cos z}\right]^{1/2}										
	\end{aligned}
\end{equation}
where $z_0=\arcsin\left(\frac{1}{m+1}\right).$  Using the estimate 
 $\displaystyle{\frac{\sinh (kz)}{\sinh(k\pi/2)}\leq e^{-k(\pi/2-z_0)}}\;$ for $z\in(0,z_0)$,  we find
\begin{equation}
	\begin{aligned}
		\hat{S}^{(l)}(k)&\leq \frac{(m+1)^{3/2}}{\sqrt{2\pi}\sqrt{m(m+2)}}\frac{1}{\sqrt{2l+1}} e^{-k(\pi/2-z_0)}\left[\int_0^{z_0}dz \frac{1}{\cos z}\right]^{1/2}\\
										&\leq \frac{(m+1)^{3/2}}{\sqrt{2\pi}\sqrt{m(m+2)}}\frac{1}{\sqrt{2l+1}} e^{-k(\pi/2-z_0)}\left[\log\left(\tan \left(\frac{\pi}{4}  +  \frac{z_0}{2}\right)\right)\right]^{1/2}\\
										&\leq \frac{(m+1)^{3/2}}{\sqrt{2\pi}\sqrt{m(m+2)}}\frac{1}{\sqrt{2l+1}} e^{-k(\pi/2-z_0)}\left[\log\left(\sqrt{1+\frac{2}{m}}\right)\right]^{1/2}
	\end{aligned}
\end{equation}
where in the last step we have used the elementary formula $\tan \frac{x}{2}=\frac{\sin x}{1+\cos x}$ and the definition of $z_0.$ 
Taking into account of the definition of $\alpha(m)$ and $\beta(m)$ given in \eqref{eq:alphabeta}, we have shown
\begin{equation}\label{eq:S_l_upper}
	\hat{S}^{(l)}(k)\leq \frac{\alpha(m)}{\sqrt{2l+1}}e^{-\beta (m)k}.
\end{equation}

\n
By equation \eqref{eq:C(m)2} and the above inequality we obtain 
\begin{equation}\label{stiaex}
	\mathcal{C}(m)  \leq \sum_{\substack{l=1,\\ l\textup{ odd}}}^{+\infty}\frac{2l+1}{2\pi}\left|\left\{k\in[0,+\infty)\,|\,\frac{\alpha(m)}{\sqrt{2l+1}}e^{-\beta (m)k}> \frac{1}{\sqrt{2\pi}}\right\}\right|.
\end{equation}
The measure of the set in the r.h.s. of \eqref{stiaex}  is different from zero only if 
\begin{equation}
	\frac{\alpha(m)}{\sqrt{2l+1}}>\frac{1}{\sqrt{2\pi}}
\end{equation}
i.e., only if  $l\leq l_0(m)$, where $l_0(m)$ is the largest odd integer smaller that $\pi \alpha(m)^2 - \f{1}{2}$.

\n
Therefore we have
\begin{equation}
	\mathcal{C}(m)\leq \sum_{\substack{l=1,\\l\textup{ odd}}}^{l_0(m)}\frac{2l+1}{2\pi}\left|\left\{k\in[0,+\infty)\,|\,\frac{\alpha(m)}{\sqrt{2l+1}}e^{-\beta (m)k}> \frac{1}{\sqrt{2\pi}}\right\}\right|.
\end{equation}

\n
For any $l \leq l_0(m)$, let
\begin{equation}
	K_l(m):=\frac{1}{\beta(m)}\log\frac{\sqrt{2\pi}\,\alpha(m)}{\sqrt{2l+1}}
\end{equation} 
be the unique positive solution  of the equation   $\displaystyle{\frac{\alpha(m)}{\sqrt{2l+1}}e^{-\beta(m) k}=\frac{1}{\sqrt{2\pi}}}$. Then 

\begin{equation}
	\begin{aligned}
		\mathcal{C}(m)&\leq 
\sum_{\substack{l=1,\\ l\textup{ odd}}}^{l_0(m)} 
\frac{2l+1}{2\pi\beta(m) }
\log\left( \frac{  \sqrt{2\pi}\, \alpha(m)  }{  \sqrt{2l+1} }\right)\\
&\leq \frac{1}{2\pi\beta(m)} \log\left( \frac{ \sqrt{2\pi}\, \alpha(m) }{ \sqrt{3}}\right)
\sum_{\substack{l=1,\\ l\textup{ odd}}}^{l_0(m)} (2l+1)\\
									&= \frac{1}{4\pi\beta(m)}
\log\left( \frac{ \sqrt{2\pi}\, \alpha(m) }{ \sqrt{3}}\right)(l_0(m)^2+3l_0(m)+2)
	\end{aligned}
\end{equation}
and this concludes the proof of \eqref{eq:C(m)_upper}.

\vs
\n
Let us prove the lower bound for $\mathcal C(m)$.
By \eqref{eq:C(m)2} we immediately get
	\begin{equation}\label{eq:low_bound}
		\mathcal{C}(m)\geq \frac{3}{2\pi} \left|\left\{k\in[0,+\infty)\,|\,\hat{S}^{(1)}(k)> \frac{1}{\sqrt{2\pi}}\right\}\right|.
	\end{equation}
	
	\n
Using \eqref{maxs} and definition~\ref{def:m*}, we find that 
\be
\hat{S}^{(1)}(0) = \f{\Lambda(m)}{\sqrt{2\pi}} >\frac{1}{\sqrt{2\pi}} \;\;\;\;\text{for}\;\; m<m_*
\ee
 and this implies strictly positivity of the right hand side of \eqref{eq:low_bound}. 	Furthermore, by monotonicity of $\hat{S}^{(1)}(k)$, we have
	\begin{equation}\label{eq:k_1(m)}
		\left|\left\{k\in[0,\infty) | \hat{S}^{(1)}(k)>\frac{1}{\sqrt{2\pi}}\right\}\right|=k_1(m)
	\end{equation}
	where $k_1(m)$  is the unique positive solution of the equation  $\hat{S}^{(1)}(k)=\frac{1}{\sqrt{2\pi}}$.  From \eqref{eq:low_bound} and \eqref{eq:k_1(m)} the lower bound \eqref{eq:C(m)} follows.

\end{proof}

\begin{remark}
Let us check that $l_0(m) \geq 1$ for $m<m_*$. Recall that $l_0(m)$ is the largest odd integer smaller than $g(m):= \pi \alpha(m)^2 - 1/2$. It is easy to see that the function $g$ is decreasing, $g(m)\to +\infty$ for $m\to 0$ and $g(m)\to 0$ for $m\to +\infty.$ Then the assertion  follows if one observes that $g(m_*) >1$ (indeed, $g(m_*)\simeq 6.65$).

		
\end{remark}

\begin{remark}
	The lower bound of $\mathcal{C}(m)$ can be improved. Recall that  $m_* :=m_1^*$ is defined as the solution of $\sqrt{2\pi} \,\hat{S}^{(1)}(0) = \Lambda(m) =1$. By analogy, for each $l$ odd we can define  a critical mass $m_l^*$ as the solution of  $\sqrt{2\pi}\, \hat{S}^{(l)}(0)=1.$  By \eqref{eq:S^l_e_o}, $m_l^*$ solves the equation 
	\begin{equation}	
		\frac{2(m+1)}{\pi\sqrt{m(m+2)}}\int_0^1dy\, P_l(y)\frac{\arcsin\left(\frac{y}{m+1}\right)}{\cos\left(\arcsin\left(\frac{y}{m+1}\right)\right)}=1.
	\end{equation}
	One can show (see \cite[appendix A]{CDFMT15}) that $m_l^*$ is uniquely defined, it is  decreasing in $l$ and $\sqrt{2\pi}\, \hat{S}^{(l)}(0)>1$ $\left(\sqrt{2\pi}\, \hat{S}^{(l)}(0)<1\right)$  if $m<m_l^*$ ($m>m_l^*$). 
	Hence, if $m_{L+2}^*<m<m_L^*$ for some $L$ then in \eqref{eq:C(m)2} each term with $l\leq L$ is certainly nonzero and this implies
	\begin{equation}
		\mathcal{C}(m)\geq \sum_{\substack{l=1\\l \textup{ odd}}}^{L} \frac{2l+1}{2\pi}\left|\left\{k\in[0,+\infty)|\hat{S}^{(l)} (k) >\frac{1}{\sqrt{2\pi}}\right\}\right|.
	\end{equation}
	Obviously, for $L=1$ we have only the first term and the inequality above  reduces to \eqref{eq:low_bound}.
\end{remark}

\begin{remark}
	It is worth noticing that, by definition~\ref{def:m*}, we have $\sqrt{2\pi}\, \hat{S}^{(1)}(0)< 1$ for $m>m_*$ which, via \eqref{eq:S^l_mon} and \eqref{maxs}, implies $\mathcal{C}(m)=\displaystyle{\lim_{z\to 0^-}\frac{N(z)}{|\log|z||}=0.}$ In fact, we conclude this section showing that for $m>m_*$ the discrete spectrum of $H$ is finite, i.e., there is no   Efimov effect.
\end{remark}
\vsa

\vs
\begin{proof}[Proof of theorem~\ref{th:main1} (case $m>m_*$)] 
For $z<0$ and $\varepsilon \in (0,1) $ the following inequality holds
	\ba\label{eq:weyl}
		n(1,\boldsymbol{A}(z)) &\leq& n(1-\varepsilon, \mathcal{S}(z)) + n\left(\frac{\varepsilon}{3},\boldsymbol{A}(z)-\boldsymbol{A}_0(z) \right)+n\left(\frac{\varepsilon}{3},\mathcal{B}(z)-\tilde{\mathcal{B}}(z)\right)
		\nonumber\\
		&+& n\left(\frac{\varepsilon}{3},\mathcal{B}_0(z)-\mathcal{S}(z)\right)
	\ea
where the operators $\boldsymbol{A}_0(z), \mathcal{B}(z), \tilde{\mathcal{B}}(z), \mathcal{S}(z)$ have been defined in the previous section.  Such an inequality is a direct consequence of lemmas~\ref{laa0}, \ref{la0b}, \ref{lB}, \ref{lS_R} and of the following technical result (see, e.g.,  \cite{BS}): if $A,B$ are compact selfadjoint operators and $\lambda_i>0,$ $i=1,2$ then
	\begin{equation}
		n(\lambda_1+\lambda_2,A+B)\leq n(\lambda_1,A)+n(\lambda_2,B).
	\end{equation}
Note that for $z\leq0$ the operators appearing in the last three terms of \eqref{eq:weyl} are compact and continuous in $z$. Therefore, the last three terms of \eqref{eq:weyl} remain finite for $z \rightarrow 0^-$.

\n
Let us consider the first term in the right hand side of \eqref{eq:weyl}, i.e., $n(1-\varepsilon,\mathcal{S}(z))$.

	\n
	We notice that
	\begin{equation}\label{eq:Sxi_z}
		\begin{aligned}
			(\mathcal{S}(z)\xi,\xi)=(\mathcal{S}\xi_z,\xi_z) = (SM \xi_z, M \xi_z)
		\end{aligned}
	\end{equation}
	where $\xi_z(p)=(\chi_{|z|^{-1/2}}-\chi_1)(p)\xi(p)$, $\mathcal{S}$  is the operator in $L^2(\R^3)$  defined by
	\begin{equation}
		(\mathcal{S} f) (p)=-\frac{1}{4\pi^2}\frac{(m+1)^2}{\sqrt{m^3(m+2)}}\int dq\,\frac{f(q)}{|p|^{1/2}\left(\frac{p^2}{2\mu}+\frac{q^2}{2\mu}+\frac{p\cdot q}{m}\right)|q|^{1/2}},
	\end{equation}
$M$ is the unitary operator     defined in \eqref{eq:M} and $S$ acts   on $L^2(\R_+\times \mathbb{S}^2,dr\otimes d\Omega)$ as follows
	\begin{equation}
		(Sf) (r,\omega)=-\frac{1}{4\pi^2}\frac{m+1}{\sqrt{m(m+2)}}\int_{-\infty}^{+\infty}d\rho\int_{\mathbb{S}^2}d\Omega(\zeta) \frac{f(\rho,\zeta)}{\cosh(r-\rho)+\frac{\omega\cdot \zeta}{m+1}}.
	\end{equation}
Let us estimate $(Sf,f)$. By \eqref{eq:f_lm} and \eqref{eq:S^l} we get
	\begin{equation}
		\begin{aligned}
			(Sf,f)&=-\frac{1}{2\pi}\frac{m+1}{\sqrt{m(m+2)}} \sum_{l=0}^{+\infty}\sum_{\nu=-l}^l \int_{-\infty}^{+\infty}\!\!\!dr\,\overline{f_{l\nu}(r)}\int_{-\infty}^{+\infty}\!\!\!d\rho\, f_{l\nu}(\rho) \int_{-1}^1dy\, \frac{P_l(y)}{\cosh(r-\rho)+\frac{y}{m+1}}\\
						&=\sum_{l=0}^{+\infty}\sum_{\nu=-l}^l \sqrt{2\pi} (\hat{f}_{l\nu},\hat{f}_{l\nu}\hat{S}^{(l)})\\
						&=\sqrt{2\pi}\sum_{l=0}^{+\infty}\sum_{\nu=-l}^l \int_{-\infty}^{+\infty}dk\, |\hat{f}_{l\nu}(k)|^2\hat{S}^{(l)}(k).
		\end{aligned}
	\end{equation}
	
	\n
Taking into account of  \eqref{eq:S^l_sign}, \eqref{eq:S^l_mon} and ~\eqref{maxs} we deduce
	\begin{equation}\label{sff}
		\begin{aligned}
			(Sf,f)&\leq \sqrt{2\pi}\sum_{\substack{l=1\\l\textup{ odd}}}^{+\infty}\sum_{\nu=-l}^l \int_{-\infty}^{+\infty}dk\, |\hat{f}_{l\nu}(k)|^2\hat{S}^{(1)}(k)\\
						&\leq \Lambda(m)
						\sum^{+\infty}_{\substack{l=1\\l\textup{ odd}}}\sum_{\nu=-l}^l \norma{f_{l\nu}}^2\\
						&\leq \Lambda(m)\norma{f}^2.
		\end{aligned}
	\end{equation}
By \eqref{eq:Sxi_z} and \eqref{sff} we have	
\be
(\mathcal S(z) \xi_z, \xi_z) \leq \Lambda(m) \|M\xi_z \|^2 \leq \Lambda(m) \|\xi\|^2
\ee
where $0<\Lambda(m)<1$ for $m>m_*$ (see definition~\ref{def:m*}).   Therefore, the operator $\mathcal S(z)$ does not have eigenvalues larger than $\Lambda(m)$, i.e.,
\be\label{ns0}
n(1-\varepsilon, \mathcal S(z))=0
\ee
for any $\varepsilon \in (0, 1-\Lambda(m))$. By theorem \ref{th:fadd} and   \eqref{eq:weyl}, \eqref{ns0} we find
\be\label{nz0}
N(z) \leq n\left(\frac{\varepsilon}{3},\boldsymbol{A}(z)-\boldsymbol{A}_0(z) \right)+n\left(\frac{\varepsilon}{3},\mathcal{B}(z)-\tilde{\mathcal{B}}(z)\right)
		+ n\left(\frac{\varepsilon}{3},\mathcal{B}_0(z)-\mathcal{S}(z)\right)
\ee
for any $z<0$ and $\varepsilon \in (0, 1-\Lambda(m))$.
	Taking the limit $z\to0^-$ and using compactness and continuity in $z\leq0$ of the operators in the right hand side of \eqref{nz0}, we obtain the finiteness of the number of negative eigenvalues of $H.$
	
\end{proof}

\vs\vs

\n
 This work is partially supported by Gruppo Nazionale
per la Fisica Matematica (GNFM-INdAM). 
The authors thank M. Correggi and D. Finco for many interesting discussions on the subject of this work.

\vs\vs\vs


\vs\vs

\end{document}